\documentclass{article}
\usepackage{amsmath}
\usepackage{amssymb}
\usepackage{amsthm}
\usepackage{hyperref}
\usepackage[margin=1.0in]{geometry}
\usepackage{multicol}
\usepackage{graphicx}
\usepackage{float}
\usepackage{mathrsfs}
\usepackage{tikz}
\usepackage{shuffle}

\newtheorem{thm}{Theorem}[section]
\newtheorem{lem}{Lemma}[section]
\newtheorem*{cor}{Corollary}

\newcommand{\SM}{\setminus}

\newcommand{\I}{\infty}

\newcommand{\af}{\alpha}
\newcommand{\bt}{\beta}
\newcommand{\gm}{\gamma}

\newcommand{\ep}{\varepsilon}

\newcommand{\et}{\eta}

\newcommand{\te}{\theta}
\newcommand{\ld}{\lambda}
\newcommand{\sm}{\sigma}

\newcommand{\ph}{\varphi}

\newcommand{\rh}{\rho}
\newcommand{\om}{\omega}
\newcommand{\ta}{\tau}

\newcommand{\Dt}{\Delta}

\newcommand{\Ld}{\Lambda}

\newcommand{\Om}{\Omega}

\newcommand{\Z}{{\mathbb{Z}}}
\newcommand{\R}{{\mathbb{R}}}

\newcommand{\id}{{\mathrm{id}}}

\newcommand{\sgn}{{\mathrm{sgn}}}

\newcommand{\ul}{\underline}
\newcommand{\mft}{\mathfrak{t}}

\newcommand{\rmn}[1]{{\rm{#1}}}

\title{The Partition Function of Log-Gases with Multiple Odd Charges}
\author{Elisha D. Wolff  \and Jonathan M. Wells}
\date{\today}
\newcommand{\keywords}{Partition function, Berezin integral, Pfaffian, Hyperpfaffian, Grand canonical ensemble, Shuffle algebra}

\begin{document}
	
	\maketitle

	\begin{abstract} We use techniques in the shuffle algebra to present a formula for the partition function of a one-dimensional log-gas comprised of particles of (possibly) different integer charges at certain inverse temperature $\bt$ in terms of the Berezin integral of an associated non-homogeneous alternating tensor. This generalizes previously known results by removing the restriction on the number of species of odd charge. Our methods provide a unified framework extending the de Bruijn integral identities from classical $\beta$-ensembles ($\beta = 1, 2, 4$) to multicomponent ensembles, as well as to iterated integrals of more general determinantal integrands.
	\end{abstract}

	
	\section{Introduction}
	\label{sec:intro}
	Suppose a finite number of charged particles are placed on an infinite wire represented by the real line. The charges of the particles are assumed to be integers with the same sign, and the particles are assumed to repel each other with logarithmic interactions. We assume any two particles of the same charge, which we will call same \emph{species}, are indistinguishable. The wire is imbued with a potential which discourages the particles from escaping to infinity in either direction, and heat is applied to the system according to a parameter we call the inverse temperature $\bt$.
	
	We consider two ensembles:
	\begin{enumerate}
		\item \emph{The Canonical Ensemble}, in which the number of particles of each species is fixed; in this case, we say fixed population.
		\item \emph{The Isocharge Grand Canonical Ensemble}, in which the sum of the charges of the particles is fixed, but the number of particles of each species is allowed to vary; in this case, we say the total charge of the system is fixed.
	\end{enumerate}
	
	In contrast, the \emph{Grand Canonical Ensemble} traditionally refers to the ensemble in which the total number of particles is not fixed. For computational purposes, it is beneficial to group configurations which share the same total charge. The Grand Canonical Ensemble is then a disjoint union (over all possible sums of charges) of our Isocharge ensembles.

	In 2012, Sinclair \cite{Sinclair2012} provided a closed form of the partition function for both ensembles in terms of Berezin integrals (see \autoref{ssec:berezin}) of alternating tensors, but only for certain $\bt$ and only for ensembles with at most one species of odd charge. Here, we provide an alternative framework which allows us to generalize the result to ensembles with arbitrary mix of odd and even charges, albeit with the same limitations on $\beta$. These Berezin integral expressions are analogous to the (Pfaffian) de Bruijn integral identities \cite{deBruijn1955} (for classical $\bt=1$ and $\bt=4$ ensembles). Our methods are analogous to algebraic methods first used by Thibon and Luque \cite{Luque2002} in the classical single species cases. In \autoref{ssec:debruijn}, we demonstrate how these methods can also be applied to iterated integrals of more general determinantal integrands. 
	
	By first conditioning on the number of particles of each species, the partition function for the Isocharge Grand Canonical Ensemble is built up from the partition functions of the Canonical type, revealing the former to be a generating function of the latter as a function of the \emph{fugacities} of each species (roughly, the probability of the occurrence of any one particle of a given charge). In \autoref{sec:circular}, we produce analogous results for charged particles placed on the unit circle in the complex plane.

	\subsection{Historical context}
	\label{sec:hist}
	
	
	The $\bt$-ensembles are a well-studied collection of random matrices whose eigenvalue densities take a common form, indexed by a non-negative, real parameter $\bt$. Suppose $\mu$ is a continuous probability measure on $\R$ with Radon-Nikodym derivative $\frac{d \mu}{d x} = w(x)$. For each $\bt\in \R_{> 0}$, consider the $N$-point process specified by the joint probability density
	\[
	\rh_N(x_1,\ldots,x_N)  =\frac{1}{Z_{N}(\bt)N!}\prod_{i<j}\left|x_j-x_i\right|^\bt \prod_i w(x_i)
	\]
	where $Z_N(\bt)$, which denotes the \emph{partition function} of $\bt$, and $Z_{N}(\bt)N!$ is the normalizing constant required for $\rh_N$ to be a probability density function. Explicitly,
	$$
Z_N(\bt) = \frac{1}{N!} \int_{\R^N} 	\prod_{i<j}\left|x_j-x_i\right|^\bt \prod_i w(x_i) \, dx_1 \dots, dx_N.
	$$
The integral that appears above is closely related to the Selberg integral \cite{Selberg1944} and its generalization, the Aomoto integral \cite{Aomoto1987}. In \cite{Luque2003} and \cite{Luque2004}, Luque and Thibon presented an evaluation of these integrals in terms of hyperdeterminants, which were first introduced by Cayley in \cite{Cayley1843}.

Moreover, the eigenvalue density function $\rho_N$ above can be identified with the Boltzmann factor of the previously discussed log-gas particles, as first observed by Dyson \cite{Dyson1962}, and further developed by Forrester in \cite{Forrester2010}. 
	
	The \emph{classical} $\bt$-ensembles are those with $\bt=1,2,4$ and $w(x) = e^{-x^2/2}$ and correspond to Hermitian matrices with real, complex, or quaternionic Gaussian entries (respectively). The $\beta = 1$ case was first investigated in the 1950s by Wigner in the context of nuclear physics \cite{wig}, following Wigner's discovery of a similar ensemble of real-valued matrices used by Wishart in the 1920s in the field of multivariate statistics \cite{wish}. In the subsequent decade, Dyson and Mehta \cite{Dyson1963} unified a previously disparate collection of random matrix models by demonstrating that the three classic $\bt$-ensembles are each variations of a single action on random Hermitian matrices (representing the three associative division algebras over $\R$). In \cite{Dumitriu2002}, Dumitriu and Edelman provide tridiagonal matrix models for $\bt$-ensembles of arbitrary positive $\bt$, which are then used by Ram\'irez, Rider, and Vir\'ag in \cite{Ramirez2011} to obtain the asymptotic distribution of the largest eigenvalue.   
	
	For each $1\leq n\leq N$, define the $n^{\rm{th}}$ correlation function by 
	\[
	R_n(x_1,\ldots,x_n)=\frac{N!}{(N-n)!}\int_{\R^{N-n}}\rh_N(x_1,\ldots,x_n,y_1,\ldots,y_{N-n})\,dy_1\cdots dy_{N-n}.
	\]
	It turns out that the correlation function for the classic $\beta$-ensembles takes a particularly nice algebraic form. For example, when $\bt=2$, it can be shown using only elementary matrix operations and Fubini's Theorem that
	\[
	R_n(x_1,\ldots,x_n)=\det(K(x_i,x_j)_{1\leq i,j\leq n}),
	\]
	where the \emph{kernel} $K(x,y)$ is a certain square integrable function $\R\times \R\to \R$ that can most easily be expressed in terms of a family of polynomials which are orthogonal with respect to the measure $\mu$. For this reason, we say the classical $\bt=2$ ensemble is an example of a \emph{determinantal} point process. The details of this derivation are given in \cite{Mehta2004}. Similarly, when $\bt=1$ or $4$, 
	\[
	R_n(x_1,\ldots,x_n)={\rm{Pf}}(K_\beta(x_i,x_j)_{1\leq i,j\leq n}),
	\]
	where ${\rm{Pf}}(A)=\sqrt{\det(A)}$ denotes the Pfaffian of an antisymmetric matrix $A$, and where $K(x,y)$ is a certain $2\times 2$ matrix-valued function whose entries are square-integrable, and which satisfies $K(x,y)^T = -K(y,x)$. We then say the classical $\bt=1$ and $\bt=4$ ensembles are examples of \emph{Pfaffian} point processes. This result was first shown for circular ensembles by Dyson in \cite{Dyson1962}, then for Gaussian ensembles by Mehta in \cite{Mehta2004} and then for general weights ($\mu$) by Mehta and Mahoux in \cite{mahoux}, except for the case $\bt=1$ and $N$ odd. Finally, the last remaining case was given by Adler, Forrester, and Nagao in \cite{adler}. An investigation of \emph{hyperdeterminantal point processes}, another generalization of the determinantal point process, tracing its roots to Cayley's hyperdeterminants, can be found in \cite{Evans2009}.
	

	\subsection{Hyperpfaffian partition functions}
	\label{sec:hpf}
	
	Derivations of the determinantal and Pfaffian expressions of the correlation functions have been presented in numerous ways over the past several decades. Of particular note is the method of Tracy and Widom \cite{wid}, who first show that the partition function is determinantal or Pfaffian, and then use matrix identities and generating functions to obtain a corresponding form for the correlation functions.	
	
	But recognizing the partition function $Z_N(\bt)$ as the determinant or Pfaffian of a matrix of integrals of appropriately chosen orthogonal polynomials is essential and nontrivial. One way to do this is to apply the Andreief determinant identity \cite{Andreief1886} to the iterated integral which defines $Z_N(\bt)$. This is immediate when $\bt=2$, and viewing the Pfaffian as the square root of a determinant, this identity can also be applied (with some additional finesse) when $\bt=1$ or $4$. However, viewing the Pfaffian in the context of the exterior algebra allows us to extend the Andreief determinant identity to analogous Pfaffian identities (referred to as the de Bruijn integral identities \cite{deBruijn1955}). 
	
	In 2002, Luque and Thibon \cite{Luque2002} used techniques in the shuffle algebra to show that when $\bt=L^2$ is an even square integer, the partition function $Z_N(\bt)$ can be written as a Hyperpfaffian of an $L$-form whose coefficients are integrals of Wronskians of suitable polynomials. Then in 2011, Sinclair \cite{Sinclair2011} used other combinatorial methods to show that the result also holds when $\bt=L^2$ is an odd square integer. 
	
	Here, we consider \emph{multicomponent} ensembles in which the joint probability density functions generally have the form
	\[
	\rh_N(x_1,\ldots,x_N)=\frac{1}{Z_{f}(\bt)N!}\prod_{i<j}\left|x_j-x_i\right|^{\bt f(i,j)}\prod_i w_i(x_i),
	\]
	where $f:\R^{2}\to\R$ specifies (possibly) different exponents for each factor in the product. This setup was first investigated for particles of charge $1$ and $2$ on complex unit circle by Forrester in \cite{Forrester1984Int} and \cite{Forrester1984Exact}, and then considered on the real line by Rider, Sinclair and Xu in \cite{Rider2010}. This model is closely associated to the eigenvalue densities for the real Ginibre ensembles as discussed by Forrester and Nagao in \cite{Forrester2007}, and then by Borodin and Sinclair in \cite{Borodin2009}. The limiting behavior for the two species model in the circular case was later studied by Shum and Sinclair in \cite{Shum2014}. A recent paper by Forrester and Li \cite{Forrester2021} extends these results further to express the skew orthogonal polynomials for classical weight functions in terms of hypergeometric polynomials.
	
	In 2012, Sinclair showed (using his same methods) the partition function $Z_f(\bt)$ has a Berezin integral expression when each $\sqrt{\bt}L_j$ is an even integer. Herein, we extend to arbitrary positive integers $\sqrt{\bt}L_j\in\Z_{>0}$ using  shuffle algebra techniques, analogous to the methods of Thibon and Luque. Note, single-component $\bt$-ensembles are a subset of multicomponent ensembles in which $f(i,j)=1$, and the Berezin integral is a generalization of the Hyperpfaffian. Thus, one consequence of this work is a new, streamlined derivation of the Hyperpfaffian partition functions for the single-component $\bt$-ensembles in which $\bt=L^2$ is any square integer, even or odd. 
	
	Though multicomponent ensembles are the focus of this volume, our methods further generalize to other ensembles. We explore many of these possibilities in a forthcoming volume. One key observation is that $\rh_N$ can be written as a determinant (typically easier than expressing $Z_N$ as a determinant). More generally, we can replace the partition function $Z_N$ (or $Z_f$) with an iterated integral of any determinant fitting certain criteria. In \autoref{ssec:debruijn}, we describe how to modify the statement of the main results of this volume, producing a vast generalization of the de Bruijn integral identities.

	
	

	\subsection{The multicomponent setup}
	\label{sec:setup}
	
	Let $J\in \Z_{> 0}$ be a positive integer, the maximum number of distinct charges in the system. Let $\vec{L}=(L_1,L_2,\ldots,L_J)\in (\Z_{> 0})^J$ be a vector of distinct positive integers which we will call the \emph{charge vector} of the system. Let $\vec{M}\in (\Z_{\geq 0})^J$ be a vector of non-negative integers which we will call the \emph{population vector} of the system. Each $M_j$ gives the number (possibly zero) of indistinguishable particles of charge $L_j$. Let 
	\[
	{\textbf{x}}=(\textbf{x}^1,\textbf{x}^2,\ldots,\textbf{x}^J)\in \R^{M_1}\times \R^{M_2}\times\cdots\times\R^{M_J}
	\]
	so that $\textbf{x}^j=(x_1^j,x_2^j,\ldots,x_{M_j}^j)\in\R^{M_j}$ for each $j$. We call ${\textbf{x}}$ the \emph{location vector} of the system in which each $x_m^j\in\R$ gives the location of a particle of charge $L_j$. We call $\textbf{x}^j$ the \emph{location vector} for the species with charge $L_j$. If some $M_j=0$, then we take $\textbf{x}^j$ to be the empty vector. 
	
	The particles are assumed to interact logarithmically on an infinite wire so that the contribution of energy to the system by two particles of charge $L_j$ and $L_k$ at locations $x_m^j$ and $x_n^k$ respectively is given by $-L_jL_k\log |x_n^k-x_m^j|$. If $U$ is the potential on the system, then at inverse temperature $\bt$, the total potential energy of the system is given by
	\[
	E_{\vec{M}}({\textbf{x}})=\bt\sum_{j=1}^JL_j\sum_{m=1}^{M_j}U(x_m^j)-\bt\sum_{j=1}^JL_j^2\sum_{m<n}\log|x_n^j-x_m^j|-\bt\sum_{j<k}L_jL_k\sum_{m=1}^{M_j}\sum_{n=1}^{M_k}\log|x_n^k-x_m^j|.
	\]
	The first type of iterated sum accounts for the potential $U$, the second type of iterated sum accounts for interactions between particles of the same charge $L_j$, and the third type of iterated sum accounts for the interactions between particles of distinct charges $L_j$ and $L_k$. 
	
	With this setup, the relative density of states (corresponding to varying location vectors ${\textbf{x}}$) is given by the Boltzmann factor
	\begin{align*}
		\Om_{\vec{M}}({\textbf{x}})&=\exp(-E_{\vec{M}}({\textbf{x}}))\\&=\prod_{j=1}^J\prod_{m=1}^{M_j}\exp\left(-\bt L_jU(x_m^j)\right)\times \prod_{j=1}^J\prod_{m<n}\left|x_n^j-x_m^j\right|^{\bt L_j^2}\times \prod_{j<k}\prod_{m=1}^{M_j}\prod_{n=1}^{M_k}\left|x_n^k-x_m^j\right|^{\bt L_jL_k}.
	\end{align*}
	Later, it will be convenient to write $W_{\vec{M}}({\textbf{x}})$ in place of the first of the three iterated products above. In the case when $\sqrt{\bt}L_j\in \Z$ for all $j$, we will also write $\left|\det V^{\vec{L},\vec{M}}({\textbf{x}})\right|$ in place of the product of the remaining two iterated products above. In \autoref{ssec:confluent}, we construct the matrix $V^{\vec{L},\vec{M}}({\textbf{x}})$ of which this is the determinant. Then the probability of finding the system in a state corresponding to a location vector ${\textbf{x}}$ is given by the joint probability density function
	\[
	\rh_{\vec{M}}({\textbf{x}})=\frac{\Om_{\vec{M}}({\textbf{x}})}{Z_{\vec{M}}M_1!M_2!\cdots M_J!}=\frac{W_{\vec{M}}({\textbf{x}})\left|\det V^{\vec{L},\vec{M}}({\textbf{x}})\right|}{Z_{\vec{M}}M_1!M_2!\cdots M_J!},
	\]
	where the \emph{partition function} (of the Canonical Ensemble) $Z_{\vec{M}}$ is the normalization constant given by
	\begin{align*}
		Z_{\vec{M}}&=\frac{1}{M_1!M_2!\cdots M_J!}\int_{\R^{M_1}}\cdots \int_{\R^{M_J}}\Om_{\vec{M}}({\textbf{x}})\,d\nu^{M_1}(\textbf{x}^1)\,d\nu^{M_2}(\textbf{x}^2)\cdots d\nu^{M_J}(\textbf{x}^J)\\&=\frac{1}{M_1!M_2!\cdots M_J!}\int_{\R^{M_1}}\cdots \int_{\R^{M_J}}\left|\det V^{\vec{L},\vec{M}}({\textbf{x}})\right|\,d\mu_1^{M_1}(\textbf{x}^1)\,d\mu_2^{M_2}(\textbf{x}^2)\cdots d\mu_J^{M_J}(\textbf{x}^J),
	\end{align*}
	with Lebesgue measure $\nu^{M_j}$ on $\R^{M_j}$ and $d\mu_j (x)=w_j(x)\,dx=\exp\left(-\bt L_jU(x)\right)dx$. 
	
	Note, the factorial denominators appear since particles of the same charge are indistinguishable, giving many different representatives for each state. In particular, the integrand is invariant under permutation of $\{x_1^j,x_2^j,\ldots,x_{M_j}^j\}$ for any $j$ fixed. At this point, it is necessary to assume the potential $U$ is one for which $Z_{\vec{M}}$ is finite. Also, replacing $\bt$ with $\bt'=\bt/b^2$ and replacing each $L_j$ with $L_j'=bL_j$ leaves $\left|\det V^{\vec{L},\vec{M}}({\textbf{x}})\right|$ unchanged. Then replacing $U$ with $U'=bU$ leaves $W_{\vec{M}}({\textbf{x}})$ unchanged. Thus, for computational purposes, we can always assume $\bt=1$. 
	
	Next, allowing the number of particles of each species to vary, let $P(\vec{M})$ be the probability of finding the system with population vector $\vec{M}$. Let $\vec{z}=(z_1,\dots,z_J)\in (\R_{> 0})^J$ be a vector of positive real numbers called the \emph{fugacity vector}. Classically, the probability $P(\vec{M})$ is given by
	\[
	P(\vec{M})=z_1^{M_1}z_2^{M_2}\cdots z_{J}^{M_J}\frac{Z_{\vec{M}}}{Z_N},
	\]
	where $Z_N$ is the partition function of the Isocharge Grand Canonical Ensemble (corresponding to fixed total charge $N$) given by
	\[
	Z_N=\sum_{\vec{L}\cdot \vec{M}=N}z_1^{M_1}z_2^{M_2}\cdots z_{J}^{M_J}Z_{\vec{M}}.
	\]
	In the above expression, the vector $\vec{L}$ of allowed charges is fixed, so we are summing over allowed population vectors $\vec{M}$. A population vector is valid only when the sum of the charges $\sum_{j=1}^JL_jM_j$ is equal to the prescribed total charge $N$. 
	
	This $Z_N$ is the primary object of interest to us. In addition to the obvious dependence on charge vector $\vec{L}$ and total charge $N$, this expression also varies with potential $U$ and inverse temperature $\bt$. The potential $U$ dictates the external forces experienced by each particle individually, affecting the measures against which we are integrating. The inverse temperature $\bt$ influences the strength of the interactions between the particles, affecting the exponents on the interaction terms in the Boltzmann factor. Taking the fugacity vector $\vec{z}$ to be a vector of indeterminants, $Z_N$ is a polynomial in these indeterminants which generates the partition functions of the Canonical Ensembles.

	Recall, each $Z_{\vec{M}}$ is an iterated integral in many variables. Our goal here is not to compute these integrals for any particular choice of several parameters. Instead, we will show, in general, how to reduce the difficulty of the computation to (sums of products of) only single or double integrals. Sinclair (2012) showed for certain values of $\bt$ (for which $\sqrt{\bt}L_j\in \Z$ for all $j$), $Z_N$ can be expressed as the Berezin integral (with respect to the volume form on $\R^N$) of the exponential of an explicit element (alternating tensor, also \emph{form}) in the exterior algebra $\bigwedge(\R^N)$ in the case when at most one of the $L_j$ is odd. In the case when exactly one of the $L_j$ is odd, say $L_1$, this required the additional restriction that $N$ be even, which would force $M_1$ to be even. We will show his expression can be extended to arbitrary $\vec{L}$ (for which any number of the $L_j$ may be odd), and in the case when $N$ is odd, we give an analogous Berezin integral expression with respect to the volume form on $\R^{N+1}$. 
	
	\section{Preliminary definitions}
	\label{sec:prelim}
	
	In this section, we introduce a mix of conventions and definitions which simplify the statement of our main results. First, for any positive integer $N$, let $\underline{N}$ denote the set $\{1,\ldots,N\}$. Assuming positive integers $K\leq N$, let $\mft:\underline{K}\nearrow\underline{N}$ denote a strictly increasing function from $\underline{K}$ to $\underline{N}$, meaning
	\[
	1\leq\mft(1)<\mft(2)<\cdots<\mft(K)\leq N.
	\]
	It will be convenient to use these increasing functions to track indices used in denoting minors of matrices and elements of exterior algebras, among other things (often in place of, but sometimes in conjunction with, permutations). For example, given an $N\times N$ matrix $V$, $V_{\mft}$ might denote the $K\times K$ minor composed of the rows $\mft(1),\ldots,\mft(K)$, taken from the first $K$ columns of $V$. More conventions related to indexing and permutations (which are relevant to the proofs) are introduced in \autoref{sec:lemmas} but are not necessary for the statement of our main results. 
	
	\subsection{Wronskians}
	
	For any non-negative integer $l$, define the $l^{\rm{th}}$ modified differential operator $D^l$ by
	\[
	D^lf(x)=\frac{1}{l!}\frac{d^lf}{dx^l},
	\]
	with $D^0f(x)=f(x)$. Define the modified Wronskian, Wr$(\vec{f},x)$, of a family, $\vec{f}=\{f_n\}_{n=1}^L$, of $L$ many sufficiently differentiable functions by 
	\[
	{\rm{Wr}}(\vec{f},x)=\det\left[D^{l-1}f_n(x)\right]_{n,l=1}^L.
	\]
	We call this the \emph{modified} Wronskian because it differs from the typical Wronskian (used in the study of elementary differential equations to test for linear dependence of solutions) by a combinatorial factor of $\prod_{l=1}^L l!$. 
	
	A \emph{complete} $N$-family of monic polynomials is a collection $\vec{p}=\{p_n\}_{n=1}^N$ such that each $p_n$ is monic of degree $n-1$. Given $\mft:\underline{L}\nearrow \underline{N}$, denote $\vec{p}_\mft=\{p_{\mft(k)}\}_{k=1}^L$. Then the (modified) Wronskian of $\vec{p}_\mft$ is given by
	\[
	{\rm{Wr}}(\vec{p}_\mft,x)=\det\left[D^{l-1}p_{\mft(k)}(x)\right]_{k,l=1}^L.
	\]

	\subsection{The Berezin integral}
	\label{ssec:berezin}
	Let $\ep_1,\ldots,\ep_N$ be a basis for $\R^N$. For any injection $\mft:\underline{K}\to \underline{N}$, let $\ep_\mft\in \bigwedge^{K}(\R^N)$ denote
	\[
	\ep_\mft=\ep_{\mft(1)}\wedge\ep_{\mft(2)}\wedge\cdots\wedge\ep_{\mft(K)}.
	\]
	Then $\{\ep_\mft\,|\,\mft:\underline{K}\nearrow\underline{N}\}$ is a basis for $\bigwedge^K(\R^N)$. In particular, $\bigwedge^N(\R^N)$ is a one-dimensional subspace we call the \emph{determinantal line}, spanned by
	\[
	\ep_{\rmn{vol}}=\ep_{\rm{id}}=\ep_1\wedge\ep_2\wedge\cdots\wedge\ep_N,
	\]
	which we call the \emph{volume form} (in $\R^N$). For each $0<n\leq N$, define $\frac{\partial}{\partial \ep_n}:\bigwedge^K(\R^N)\to \bigwedge^{K-1}(\R^N)$ on basis elements by 
	\[
	\frac{\partial}{\partial \ep_n}\ep_{\mft}=\begin{cases}
		(-1)^{k}\ep_{\mft(1)}\wedge\cdots\wedge\ep_{\mft(k-1)}\wedge\ep_{\mft(k+1)}\wedge\cdots\wedge\ep_{\mft(K)} & \text{ if } k=\mft^{-1}(n) \\ 0 & \text{ otherwise}
	\end{cases},
	\]
	and then extend linearly. If $n\in \mft(\underline{K})$, meaning $\ep_n$ appears as a factor in $\ep_{\mft}$, then $\frac{\partial \ep_{\mft}}{\partial \ep_n}$ is the result of permuting $\ep_n$ to the front and then removing it, picking up a sign associated with changing the order in which the basis elements occur. If $\ep_\mft$ does not have $\ep_n$ as a factor, then $\frac{\partial \ep_{\mft}}{\partial \ep_n}=0$. Given an injection $\mathfrak{s}:\underline{L}\to \underline{N}$, we define the Berezin integral \cite{Berezin1966} (with respect to $\ep_{\mathfrak{s}}$) as a linear operator $\bigwedge(\R^N)\to \bigwedge(\R^N)$ given by
	\[
	\int\ep_{\mft}\,d\ep_{\mathfrak{s}}=\int \ep_{\mft}\,d\ep_{\mathfrak{s}(1)}\,d\ep_{\mathfrak{s}(2)}\cdots d\ep_{\mathfrak{s}(L)}=\frac{\partial}{\partial \ep_{\mathfrak{s}(L)}}\cdots \frac{\partial}{\partial \ep_{\mathfrak{s}(2)}}\frac{\partial}{\partial \ep_{\mathfrak{s}(1)}}\ep_{\mft}.
	\]
	Our main results are stated in terms of Berezin integrals with respect to the volume form $\ep_{\rmn{vol}}\in \bigwedge^N(\R^N)$. Note, if $\ep_{\mft}\in \bigwedge^K(\R^N)$ for any $K<N$, then 
	\[
	\int \ep_{\mft}\,d\ep_{\rmn{vol}}=0
	\]
	because $\ep_{\mft}$ is missing some $\ep_k$ as a factor. Thus, the Berezin integral with respect to $\ep_{\rmn{vol}}$ is a projection operator $\bigwedge(\R^N)\to \bigwedge^N(\R^N)\cong \R$. In particular, if $\sm\in S_N$, then
	\[
	\int \ep_{\sm}\,d\ep_{\rmn{vol}}=\sgn (\sm).
	\]

	\subsection{Exponentials of forms}

	For $\om\in \bigwedge(\R^N)$ and positive integer $m$, we write
	\[
	\om^{\wedge m}=\om\wedge\cdots\wedge \om,
	\]
	with $\om$ appearing as a factor $m$ times. By convention, $\om^{\wedge 0}=1$. We then define the exponential
	\[
	\exp(\om)=\sum_{m=0}^\I\frac{\om^{\wedge m}}{m!}.
	\]
	Moreover, suppose $\om=\om_1+\om_2+\cdots+\om_J$ where each $\om_j\in \bigwedge^{L_j}(\R^N)$ and each $L_j$ even, then (we say each $\om_j$ is a homogeneous even form of length $L_j$ and) it is easily verified 
	\[
	\exp({\om})=\exp({\om_1+\cdots+\om_J})=\exp({\om_1})\wedge\cdots\wedge \exp({\om_J}).
	\]
	
	In our main results, the forms we are exponentiating are non-homogeneous. However, we get a homogeneous form in the case when we only have one species of particle. In that case, exactly one summand in the exponential will live at the determinantal line. Assuming $\om\in \bigwedge^L(\R^N)$ with $LM=N$, we get
	\[
	\int \exp(\om)\,d\ep_{\rmn{vol}}=\int \sum_{m=0}^\I\frac{\om^{\wedge m}}{m!}\,d\ep_{\rmn{vol}}=\int \frac{\om^{\wedge M}}{M!}\,d\ep_{\rmn{vol}}={\rm{PF}}(\om),
	\]
	where PF$(\om)$ is the \emph{Hyperpfaffian} of $\om$, the real number coefficient on $\ep_{\rmn{vol}}$ in $\frac{\om^{\wedge M}}{M!}$. Thus, this Berezin integral is the appropriate generalization of the Hyperpfaffian. To avoid confusing this Berezin integral with other integrals which appear in our computations, we'll write
	\[
	{\rm{BE}}_{\rmn{vol}}(\om)=\int\exp(\om)\,d\ep_{\rmn{vol}},
	\]
	where the subscript on the left hand side indicates which form we are integrating with respect to.
	
	The partition function of an ensemble with a single species has been shown to have a Hyperpfaffian expression (for certain $\bt$) \cite{Sinclair2011}. As we generalize to ensembles with multiple species (and therefore non-homogeneous forms), we replace the Hyperpfaffian with the more general Berezin integral (of an exponential). 
	
	\section{Statement of results}
	\label{sec:results}
	
	Recall the setup from \autoref{sec:setup}. Let $\vec{p}$ be a complete $N$-family of monic polynomials. Define
	\[
	\gm_j=\sum_{\mft:\underline{L_j}\nearrow \underline{N}}\int_{\R}{\rm{Wr}}(\vec{p}_{\mft},x)\,d\mu_j(x)\,\,\ep_{\mft},
	\]
	and define
	\[
	\et_{j,k}=\sum_{\mft:\underline{L_j}\nearrow\underline{N}}\sum_{\mathfrak{s}:\underline{L_k}\nearrow\underline{N}}\int\int_{x<y}{\rm{Wr}}(\vec{p}_{\mft},x){\rm{Wr}}(\vec{p}_{\mathfrak{s}},y)\,d\mu_j(x)d\mu_k(y)\,\,\ep_{\mft}\wedge\ep_{\mathfrak{s}}.
	\]

	\begin{thm}\label{thm:alleven}  If all $L_j$ are even, then 
		\[
		Z_{N}={\rm{BE}}_{\rmn{vol}}\left(\sum_{j=1}^Jz_j\gm_j\right).
		\]
	\end{thm}
	
	The above theorem is Sinclair's (2012) for which we give a different proof and then the following generalization:
	
	\begin{thm}\label{thm:reven}If the first $r$ many $L_j$ are even and $N$ is even, then
		\[
		Z_N={\rm{BE}}_{\rmn{vol}}\left(\sum_{j=1}^rz_j\gm_j+\sum_{j=r+1}^{J}\sum_{k=r+1}^{J}z_jz_k\et_{j,k}\right).
		\]
	\end{thm}
	
	\begin{thm}\label{thm:revenodd} If the first $r$ many $L_j$ are even and $N$ is odd, then
		\[
		Z_N={\rm{BE}}_{\rmn{vol}_1}\left(\sum_{j=1}^rz_j\gm_j+\sum_{j=r+1}^{J}\sum_{k=r+1}^{J}z_jz_k\et_{j,k}+\sum_{j=r+1}^{J}z_j\gm_j\wedge \ep_{N+1}\right),
		\]
		where ${\rm{BE}}_{\rmn{vol}_1}$ includes the Berezin integral with respect to $\ep_{\rmn{vol}_1}=\ep_{\rmn{vol}}\wedge \ep_{N+1}\in \bigwedge^{N+1}(\R^{N+1})$. 
	\end{thm}
	
	For our methods, it is necessary to extend the basis by $\ep_{N+1}$ so that the new volume form $\ep_{\rmn{vol}_1}$ has even length $N+1$. More generally, we can write  
	\[
	\ep_{\rmn{vol}_k}=\ep_{\rmn{vol}}\wedge\xi_k=\ep_{\rmn{vol}}\wedge \ep_{N+1}\wedge \ep_{N+2}\wedge \cdots \wedge \ep_{N+k}.
	\]
	Then for any $\om\in\bigwedge(\R^N)\leq \bigwedge(\R^{N+k})$, we have
	\[
	\int\om\,d\ep_{\rmn{vol}}=\int\om\wedge \ep_{N+1}\wedge\cdots \wedge\ep_{N+k}\,\,d\ep_{\rmn{vol}}\,d\ep_{N+1}\cdots d\ep_{N+k}=\int\om\wedge \xi_{k}\,\,d\ep_{\rmn{vol}_k}.
	\]
	Thus, we can embed any Berezin integral computation in a higher dimension as desired.

	\begin{cor}In the single species case ($N$ indistinguishable particles of charge $L$), we get the known Hyperpfaffian expression (Sinclair 2011):
		\[
		Z_N={\rm{BE}}_{\rmn{vol}_k}(\om)={\rm{PF}}(\om),
		\]
		where $\om$ and $k$ depends on $N$ and $L$.
		\begin{enumerate}
			\item If $L$ is even, then $\om=\gm_1$ and ${\rm{BE}}_{\rmn{vol}_k}={\rm{BE}}_{\rmn{vol}}$.
			\item If $L$ is odd and $N$ is even, then $\om=\et_{1,1}$ and ${\rm{BE}}_{\rmn{vol}_k}={\rm{BE}}_{\rmn{vol}}$.
			\item If $L$ is odd and $N$ is odd, then $\om = \et_{1,1}+\gm_1\wedge \xi_L$ and ${\rm{BE}}_{\rmn{vol}_k}={\rm{BE}}_{\rmn{vol}_L}$.
		\end{enumerate}
	\end{cor}
	
	Note, we extend by $\xi_L$ instead of just $\xi_1=\ep_{N+1}$ in case 3 only so that $\gm_1\wedge\xi_L$ is a $2L$-form and therefore $\om$ is homogeneous. Every choice of $k$ produces a different but equally valid Berezin integral expression. We obtain the (Pfaffian) de Bruijn integral identities for classical $\bt=1$ and $\bt=4$ when $L=1$ and $L=2$, respectively.

	To prove these theorems, we start by giving general methods (in the shuffle algebra) for manipulating iterated integrals of determinantal integrands in \autoref{sec:lemmas}. The conclusion of these methods is a generalization of the de Bruijn integral identities, given in \autoref{ssec:debruijn}. In \autoref{sec:canonical}, we apply these identities first to $Z_{\vec{M}}$, the partition function of the Canonical Ensemble with arbitrary but fixed population $\vec{M}$. In \autoref{sec:isocharge}, we sum over all possible population vectors $\vec{M}$ to obtain $Z_N$, the partition function of the Isocharge Grand Canonical Ensemble. 
	
	\subsection{Generalized de Bruijn identities}
	\label{ssec:debruijn}
	
	Let $N=L_1+\cdots+L_J$. Define $K_j=\sum_{k=1}^{j}L_k$. Let $B(\vec{x})$ be an $N\times N$ matrix whose entries are single variable integrable functions of variables $\vec{x}=(x_1,\ldots,x_J)$. Explicitly, the first $L_1$ many columns are functions of $x_1$, the second $L_2$ many columns are functions of $x_2$, and so on up through $x_J$. For $\mft:\ul{L_j}\nearrow\ul{N}$, let $B_\mft(x_j)$ denote the $L_j\times L_j$ minor of $B(\vec{x})$ given by
	\[
	B_\mft(x_j)=\left[B(\vec{x})_{\mft(l),n+K_{j-1}}\right]_{l,n=1}^{L_j},
	\]
	equivalently obtained from $B(\vec{x})$ by taking the rows $\mft(1),\ldots,\mft(L_j)$ from the $L_j$ many columns in the same variable $x_j$. Define
	\[
	\gm_j^B=\sum_{\mft:\ul{L_j}\nearrow \ul{N}}\int_{\R}\det B_\mft (x_j)\,dx_j\,\ep_{\mft},
	\]
	and define 
	\[
	\et_{j,k}^B=\sum_{\mft:\ul{L_j}\nearrow \ul{N}}\sum_{\mathfrak{s}:\ul{L_k}\nearrow \ul{N}}\int\int_{x_j<x_k}\det B_\mft (x_j)\cdot \det B_{\mathfrak{s}} (x_k)\,dx_j\,dx_k\,\ep_{\mft}\wedge \ep_{\mathfrak{s}}.
	\]
	\begin{thm}\label{thm:debruijngen} Suppose the first $r$ many $L_j$ are even, then
		\[
		\int_{-\I<x_1<\ldots<x_J<\I}\det B(\vec{x})\,dx_1\cdots dx_J=\int \om\,d\ep_{\rmn{vol}},
		\]
		where $\om$ is defined as follows:
		\begin{enumerate}
			\item If $N$ is even, then
			\[
			\om=\frac{1}{\left(r+\frac{J-r}{2}\right)!}\bigwedge_{j=1}^{r}\gm_j^B\wedge \bigwedge_{m=1}^{(J-r)/2}\et_{r+2m-1,r+2m}^B.
			\]
			\item If $N$ is odd, then
			\[
			\om=\frac{1}{\left(r+1+\frac{J-r-1}{2}\right)!}\bigwedge_{j=1}^{r}\gm_j^B\wedge \bigwedge_{m=1}^{(J-r-1)/2}\et_{r+2m-1,r+2m}^B\wedge \gm_J^B.
			\]
		\end{enumerate}
	\end{thm}
	
	For $1\leq j\leq r$, $L_j$ is even, and $\gm_j^B$ is an even $L_j$-form. For the $L_j$ which are odd, $\et_{j,k}^B$ combines minors of odd $L_j\times L_j$ dimensions with minors of odd $L_k\times L_k$ dimensions to produce an even $(L_j+L_k)$-form. In case 1, the requirement that $N$ be even means there are an even number of odd $L_j$ to be paired down into $(J-r)/2$ pairs. In case 2, there are an odd number of odd $L_j$, so $\gm_J^B$ remains as an odd $L_J$-form. 
	
	Compared to Theorems \ref{thm:reven} and \ref{thm:revenodd}, $\det B(\vec{x})$ takes the place of the Boltzmann factor $\Om_{\vec{M}}(\textbf{x})$, so that $\det B_\mft(x_j)$ and $\det B_{\mathfrak{s}}(x_k)$ take the places of ${\rm{Wr}}(\vec{p}_\mft,x)$ and ${\rm{Wr}}(\vec{p}_{\mathfrak{s}},y)$, respectively. Recall (from \autoref{sec:setup}), the partition function $Z_N$ of the Isocharge Grand Canonical Ensemble is a linear combination of the partition functions $Z_{\vec{M}}$ of the Canonical Ensembles, with the sum taken over all possible populations $\vec{M}$. For each $\vec{M}$, we get a different Boltzmann factor integrand $\Om_{\vec{M}}(\textbf{x})$. Moreover, the defining integral is taken over all of $\R^N$. The wedge products which appear in Theorem \ref{thm:debruijngen} should be thought of as just one of the summands which appears in the expansion of the exponentials which appear in Theorems \ref{thm:reven} and \ref{thm:revenodd}.
	
	Note, we assume the functions which make up $B(\vec{x})$ are suitably integrable so that all integrals which appear in $\gm_j^B$ and $\et_{j,k}^B$ are finite. However, we do not assume any resemblance between the $L_j$ many columns in $x_j$ and the $L_k$ many columns in $x_k$. Assuming some additional consistency, we obtain a Hyperpfaffian analogue of the de Bruijn integral identities.
	
	\begin{cor}\label{cor:debruijnpf} Suppose $L_1=\cdots=L_J=L$. Under the additional assumption that $\gm_j^B=\gm$ for all $j$, and $\et_{j,k}^B=\et$ for all $j,k$ (typically because the entries of $B(\vec{x})$ in one variable $x_j$ are the same as the entries in any other variable $x_k$), 
		\[
		\int_{-\I<x_1<\ldots<x_J<\I}\det B(\vec{x})\,dx_1\cdots dx_J={\rm{BE}}_{\rmn{vol}_k}(\om)={\rm{PF}}(\om),
		\]
		where $\om$ and $k$ depend on $M$ and $L$.
		\begin{enumerate}
			\item If $L$ is even, then $\om=\gm$ and ${\rm{BE}}_{\rmn{vol}_k}={\rm{BE}}_{\rmn{vol}}$.
			\item If $L$ is odd and $M$ is even, then $\om=\et$ and ${\rm{BE}}_{\rmn{vol}_k}={\rm{BE}}_{\rmn{vol}}$.
			\item If $L$ is odd and $M$ is odd, then $\om = \et+\gm\wedge \xi_L$ and ${\rm{BE}}_{\rmn{vol}_k}={\rm{BE}}_{\rmn{vol}_L}$.
		\end{enumerate}
	\end{cor}
	
	Further generalization in Theorem \ref{thm:debruijngen} is still possible if desired. Suppose instead the first $L_1$ columns of $B(\vec{x})$ are made up of functions, not necessarily single variable, of variables $x_1,\ldots,x_a$, and the next $L_2$ columns of $B(\vec{x})$ are made up of functions of variables $x_{a+1},\ldots,x_b$. If $L_1$ is even, then we replace $\gm_1^B$ with
	\[
	\gm_1^B=\sum_{\mft:\ul{L_1}\nearrow\ul{N}}\int_{-\I<x_1<\cdots<x_a<\I}\det B_{\mft}(x_1,\ldots,x_a)\,dx_1\cdots dx_a\,\ep_{\mft},
	\]
	which now features iterated integrals of the multivariate minors. If $L_1$ and $L_2$ are odd, then we replace $\et_{1,2}^B$ with 
	\[
	\et_{1,2}^B=\sum_{\mft:\ul{L_1}\nearrow \ul{N}}\sum_{\mathfrak{s}:\ul{L_2}\nearrow \ul{N}}\int_{-\I<x_1<\cdots<x_b<\I}\det B_\mft (x_1,\ldots,x_a)\cdot \det B_{\mathfrak{s}} (x_{a+1},\ldots,x_b)\,dx_1\cdots dx_b\,\ep_{\mft}\wedge \ep_{\mathfrak{s}}.
	\]
	In general, we integrate the minors with respect to whichever variables appear with respect to the same total order on the domain induced by the original integral of $\det B(\vec{x})$.

	\section{Algebraic lemmas}
	\label{sec:lemmas}
	
	For any injection $\mathfrak{t}:\ul{K}\to \ul{N}$, let $Q_{\mathfrak{t}}$ denote $Q_{\mathfrak{t}(1),\ldots,\mathfrak{t}(K)}$ whenever it is clear from context $Q$ admits $K$ many indices, and let $Q_\mft$ denote $\{Q_{\mft(k)}\}_{k=1}^K$ whenever it is clear from context $Q$ admits only one index. For any permutation $\sm\in S_K$, we can view $\sm:\ul{K}\to\ul{K}$ as a bijection and then write $Q_{\mathfrak{t}\circ \sm}$ to denote $Q_{\mathfrak{t}\circ\sm(1),\ldots,\mathfrak{t}\circ\sm(K)}$ or $\{Q_{\mft\circ\sm(k)}\}_{k=1}^K$ as appropriate in context.

	For example, when $V$ is a matrix, $V_\mft$ is a minor. We should think of $V_\mft$ as a single object with $K$ many indices (which indicate a choice of $K$ many rows $\mft(1),\ldots,\mft(K)$ from which our minor is constructed). Similarly, if $\om\in \bigwedge^K(\R^N)$, then $A_\mft$ might denote the coefficient of $\ep_\mft$ (equivalently, an entry in a $K$-dimensional hyper array). In contrast, if $\vec{f}=\{f_k\}_{k=1}^N$ is a family of functions, then $\vec{f}_\mft=\{f_{\mft(k)}\}_{k=1}^K$ is a subfamily of $K$ functions, each indexed by a single integer. In our statement of our main results, we use these increasing function subscripts in both ways, but it is clear from context how these subscripts should be applied differently to different objects.

	Let $\mathfrak{t}:\ul{K}\nearrow\ul{N}$ denote a strictly increasing function from $\ul{K}$ to $\ul{N}$. Note, every injection $\mathfrak{s}:\ul{K}\to \ul{N}$ can be written uniquely as $\mathfrak{s}=\mathfrak{t}\circ \sm$ for some $\mathfrak{t}:\ul{K}\nearrow\ul{N}$ and $\sm\in S_K$. For any $\mathfrak{t}:\ul{K}\nearrow\ul{N}$, there exists a unique complementary $\mft':\ul{N-K}\nearrow\ul{N}$ with $\mft(\ul{K})\cup\mft'(\ul{N-K})=\ul{N}$. Define $\sgn(\mft)$ to be the signature of the permutation $\sm\in S_N$ given by 
	\[
	\sm(k)=\begin{cases}
		\mft(k) & \text{ if } k\in \ul{K} \\
		\mft'(k-K) & \text{ if } k\in \ul{N}\SM\ul{K}
	\end{cases}.
	\]
	Equivalently, 
	\[
	\sgn(\mft)=\int \ep_\mft\wedge \ep_{\mft'}\,d\ep_{\rmn{vol}}.
	\]
	For any $\Ld=(\ld_1,\ldots,\ld_K)$ which partitions $N$ and $\mft:\ul{N}\to\ul{M}$, write $\mft=(\mathfrak{t}_1\vert\ldots\vert\mft_{K})$ to indicate a decomposition of $\mft$ in which $\mft_1$ is the restriction of $\mft$ to the first $\ld_1$ positive integers, and each $\mft_{k}$ is the restriction of $\mft$ to the next $\ld_k$ positive integers. For convenience, we will treat each $\mft_k$ as having domain $\ul{\ld_k}$ instead of the appropriate subset of $\ul{N}$ of size $\ld_k$. 
	
	Conversely, for any $\mft_1:\ul{\ld_1}\to\ul{M},\ldots,\mft_{K}:\ul{\ld_K}\to \ul{M}$, we can construct $(\mathfrak{t}_1\vert\ldots\vert\mft_{K})=\mft:\ul{N}\to\ul{M}$ by defining for each $n\in\ul{N}$, $\mft(n)=\mft_k\left(n-\sum_{j=1}^{k-1}\ld_j\right)$ where this $k$, which depends on $n$, is the largest $k$ for which the difference inside the parentheses is positive. As before, it will be convenient to identify the restrictions of this new $\mft$ with the original $\mft_1,\ldots,\mft_{K}$ even though the domains aren't exactly the same. 
	
	In general, this is just a bookkeeping device which gives us a convenient notation for a choice of indices from the codomain $\underline{M}$. For example, when $\sm:\underline{N}\to \underline{N}$ is a permutation, we can think of $\sm$ as sorting $N$ many possible indices into $K$ blocks of different sizes $\ld_1,\ldots,\ld_K$ as specified by the images of the $\sm_k$. 
	
	\subsection{Decomposition of the symmetric group}

	For any $\Ld=(\ld_1,\ldots,\ld_K)$ which partitions $N$, let ${\rm{H}}(\Ld)\subseteq S_N$ denote the \emph{Young subgroup}, meaning ${\rm{H}}(\Ld)\cong S_{\ld_1}\times\cdots\times S_{\ld_{K}}$. Viewing $\sm$ as a function from $\underline{N}$ to $\underline{N}$, we can write the decomposition (with respect to $\Ld$) as $\sm=(\sm_1|\ldots |\sm_K)$. Then we should think of these $\sm_k$ belonging to the appropriate $S_{\ld_k}$.
	
	Let Sh$(\Ld)\subseteq S_N$ denote the subset of \emph{shuffle permutations}. These are permutations which satisfy $\sm(i)<\sm(j)$ whenever 
	\[
	\ld_1+\cdots+\ld_k< i<j\leq \ld_1+\cdots+\ld_{k+1}.
	\]
	Each shuffle permutation represents a way to iteratively riffle shuffle $K$ stacks of $\ld_1,\ldots,\ld_K$ many cards into a single pile of $N$ cards while preserving the original ordering within each of the $K$ stacks. Equivalently, the shuffle permutations are the $\sm\in S_N$ for which each $\sm_k$ is a strictly increasing function from $\underline{\ld_k}$ to $\underline{N}$. 
	
	Let ${\rm{Sh}}^\circ(\Ld)\subseteq {\rm{Sh}}(\Ld)$ denote the subset of \emph{ordered} shuffle permutations. These are shuffle permutations which also satisfy 
	\[
	\sm(1)<\sm(\ld_1+1)<\sm(\ld_1+\ld_2+1)<\cdots<\sm(\ld_1+\cdots +\ld_{K-1}+1).
	\]
	Using the decomposition $\sm=(\sm_1|\ldots|\sm_K)$, we can conveniently rewrite the above condition as
	\[
	\sm_1(1)<\sm_2(1)<\cdots<\sm_K(1).
	\]
	Let Bl$(\Ld)\subseteq S_N$ denote the subset of \emph{block permutations}. A block permutation represents shuffling a deck of cards by first separating the deck into a pile of the first $\ld_1$ cards, a pile of the second $\ld_2$ cards, and so on, then reassembling the deck without interlacing the piles or shuffling within any of the piles. Using the decomposition $\sm=(\sm_1|\ldots|\sm_K)$, block permutations are permutations for which each $\sm_k:\underline{\ld_k}\to \underline{N}$ acts by $\sm_k(j)=(j-1)+\sm_k(1)$.
	
	Clearly, any block permutation is determined entirely by the action on the first element of each block (of $\ld_k$ elements), of which there are $K$ many. Given a block permutation $\sm$, define $\te_\sm\in S_K$ to be the unique permutation for which $\te_\sm(i)<\te_\sm(j)$ if and only if $\sm_i(1)<\sm_j(1)$. Heuristically, $\sm$ moves blocks of $\ld_k$ consecutive elements together. $\te_\sm$ is the unique action on the $K$ many blocks determined by $\sm$. The map $\sm\mapsto\te_\sm$ is bijective, so ${\rm{Bl}}(\Ld)\cong S_K$.
	
	Note, because the blocks have (possibly) different sizes $\ld_j$, block permutations do not (in general) preserve partitions. Define 
	\[
	\Ld^\sm=(\ld_{\te_\sm^{-1}(1)},\ldots,\ld_{\te_\sm^{-1}(K)}),
	\]
	obtained from $\Ld$ by reordering its entries according to $\te_\sm$. 
	
	In the sequel, we make use of the following Lemma which decomposes permutations in the symmetric group as products of ordered shuffles, block permutations, and young permutations (composed in the opposite order).
	
	\begin{lem} \label{decomp} Let $\Ld=(\ld_1,\ldots,\ld_K)$ be a partition of $N$. Given any $\ph \in S_N$, there exists unique permutations $\ta\in {\rm{H}}(\Ld)$, $\pi\in {\rm{Bl}}(\Ld)$, and $\sm\in {\rm{Sh}}^{\circ}(\Ld^\pi)$ so that $\ph=\sm\circ \pi\circ \ta$. 
	\end{lem}
	For completeness, we include a proof of this lemma in \autoref{sec:appendix}.

	\subsection{Chen's lemma}
	
	For a set $X$ and a ring $R$, let $R\langle X\rangle$ denote the free unital algebra on $X$ over $R$. Given $u=u_1\cdots u_k$ and $u'=u_{k+1}\cdots u_n$, define an operation $\shuffle$ on $R\langle X\rangle$ as follows:
	\[
	u\shuffle u'=\sum_{\sm\in{\rm{Sh}}(k,n-k)}u_{\sm^{-1}(1)}\cdots u_{\sm^{-1}(n)}
	\]
	and by $e\shuffle e = e$ for the empty word $e \in R\langle X\rangle$. Denote by $R\langle X\rangle_{\shuffle}$ the algebra $R\langle X\rangle$ with (shuffle) multiplication $\shuffle$, which is called the Shuffle Algebra on $X$. The shuffle product was first introduced by Eilenberg and Mac Lane in \cite{Eilenberg1953}. 
	
	Let $\mathcal{H}$ be the Hilbert space $L^2(\R)$ of square integrable functions with respect to Lebesgue measure on $\R$, and suppose $H$ is a finite-dimensional subspace of $\mathcal{H}$ with basis $X$. We assume $H$ is large enough to contain all functions of interest to us.

	For any $\sm\in S_k$, let $\Dt_k(\sm)\subset(a,b)^k$ denote the region where $a<x_{\sm^{-1}(1)}<\cdots <x_{\sm^{-1}(k)}<b$. For each non-negative integer $k$, define a linear functional $\langle\cdot\rangle_k$ on the $k^{\rm{th}}$ graded component of $\R\langle X\rangle\cong T(V)$ by
	\[
	\langle f_1\otimes \cdots\otimes  f_k\rangle_k=\int_{\Dt_k(\id)}f_1(x_1)\cdots f_K(x_k)\,dx_1\cdots dx_k.
	\]
	Note, $T(V)$ denotes the tensor algebra to which $\R\langle X\rangle$ is isomorphic. Though not strictly necessary, we write $f_1\otimes \cdots\otimes f_k\in T(V)$ instead of $f_1\cdots f_k\in \R\langle X\rangle$ to avoid confusing concatenation with function multiplication.
	
	This collection $\{\langle\cdot\rangle_k\}_{k=0}^\I$ defines a functional $\langle\cdot\rangle$ on $\R\langle X\rangle$ whereby $\langle\cdot\rangle$ acts as $\langle\cdot\rangle_k$ on the $k^{\rm{th}}$ graded component of a non-homogeneous tensor. The following lemma, due to Chen \cite{Chen1953}, asserts that this operator $\langle\cdot\rangle$ is actually an algebra homomorphism from $T(V)_\shuffle$ to $\R$:
	
	\begin{lem}[Chen] \label{chen} If $f,g\in \R\langle X\rangle$, then $\langle f\shuffle g\rangle=\langle f\rangle\langle g\rangle$.
	\end{lem}
	A major hurdle in computing the partition function $Z_{\vec{M}}$ is the absolute value inside the integral. We will remove the absolute value by decomposing the domain of integration into these totally ordered subsets. However, when we change the domain of integration, we lose the ability to use Fubini's Theorem. Chen's Lemma serves the role of Fubini's Theorem, provided we can demonstrate the integrand to have the appropriate form.
	
	\begin{proof}
		We can assume $f=f_1\otimes \cdots \otimes f_k$ is a pure tensor of length $k$ and $g=f_{k+1}\otimes\cdots\otimes f_{n}$ is a pure tensor of length $n-k$ (because $\langle\cdot\rangle$ is linear and $\shuffle$ distributes over addition). Then $\langle f\shuffle g\rangle=\langle f\shuffle g\rangle_n$ is an $n$-fold iterated integral over $\Dt_n(\id)$. Likewise, $\langle f\rangle\langle g\rangle=\langle f\rangle_k\langle g\rangle_{n-k}$ is the product of a $k$-fold iterated integral over $\Dt_k(\id)$ and an $(n-k)$-fold iterated integral over $\Dt_{n-k}(\id)$. Under the integrability criteria on $X$, Fubini's Theorem tells us this is the same as an $n$-fold iterated integral over the product space $\Dt_k(\id)\times \Dt_{n-k}(\id)$. 
		
		Key to the proof of Chen's Lemma is the observation that
		\[
		\Dt_k(\id)\times \Dt_{n-k}(\id)=Z\cup \bigcup_{\sm\in {\rm{Sh}}(k,n-k)}\Dt_n(\sm),
		\]
		where $Z\subset (a,b)^n$ has measure $0$. To this end, let $Y$ be the set of points in $(a,b)^n$ whose coordinates are all distinct, and then let $Z=((a,b)^n\SM Y)\cap (\Dt_k(\id)\times \Dt_{n-k}(\id))$. If $\vec{x}\in \left(\Dt_k(\id)\times \Dt_{n-k}(\id)\right)\SM Z\subset Y$, then its coordinates are all distinct and can be ordered according to some permutation, meaning $\vec{x}\in \Dt_n(\sm)$ for a unique $\sm\in S_n$. Since $\vec{x}\in \Dt_k(\id)\times \Dt_{n-k}(\id)$, this permutation is one which separately preserves the relative order of the first $k$ coordinates and the relative order of the remaining $n-k$ coordinates, meaning $\sm\in {\rm{Sh}}(k,n-k)$. 
		Thus,
		\begin{align*}
			\langle f\rangle\langle g\rangle&=\int_{\Dt_k(\id)\times \Dt_{n-k}(\id)}f_1(x_1)\cdots f_n(x_n)\,dx_1\cdots dx_n\\&=\sum_{\sm\in{\rm{Sh}}(k,n-k)}\int_{\Dt_n(\sm)}f_1(x_1)\cdots f_n(x_n)\,dx_1\cdots dx_n.
		\end{align*}
		By relabeling the variables as $x_j=y_{\sm(j)}$, we can rewrite this sum as
		\begin{align*}
			&=\sum_{\sm\in{\rm{Sh}}(k,n-k)}\int_{\Dt_n(\id)}f_{1}(y_{\sm(1)})\cdots f_{n}(y_{\sm(n)})\,dy_1\cdots dy_n\\&=\sum_{\sm\in{\rm{Sh}}(k,n-k)}\int_{\Dt_n(\id)}f_{\sm^{-1}(1)}(y_1)\cdots f_{\sm^{-1}(n)}(y_n)\,dy_1\cdots dy_n\\&=\langle f\shuffle g\rangle.
		\end{align*}
	\end{proof}
	
	Until \autoref{sec:circular}, we will only need $(a,b)=(-\I,\I)=\R$, but we have shown Chen's Lemma to hold for more general intervals. Also, we do not need our functions to be real-valued. We just need the codomain to be associative and commutative, provided our functions are appropriately integrable. In particular, we will use complex-valued functions on real domain $[0,2\pi)$ for the circular ensembles in \autoref{sec:circular}.

	\subsection{Exterior shuffle algebra}

	Let $\{Q_{\mft}^j\,|\,j\in\ul{J}, \mft:\ul{L_j}\to\ul{N}\}$ be a subset of an alphabet $I$, and let $F$ be a field of characteristic 0. We use single integer superscripts to emphasize there are different $Q$'s which admit different numbers of integer subscripts, chosen by $\mft$'s of different sizes. For each $\mft:\ul{L_j}\to \ul{N}$, define $A_{\mft}^j\in F\langle I\rangle$ by 
	\[
	A_{\mft}^j=\sum_{\ta\in S_{L_j}}\sgn(\ta)Q_{\mft\circ \ta}^j.
	\]
	We call this the \emph{antisymmetrization} of $Q_\mft^j$. We should think of this as analogous to taking a pure tensor in the tensor algebra and then adding to it all possible orderings of the basis vectors with signs. Let $V$ be a rank $N$ free module over $R=F\langle I\rangle_\shuffle$. Define $\af_j\in \bigwedge_RV$ by 
	\[
	\af_j=\sum_{\mft:\ul{L_j}\to\ul{N}}Q_\mft^j\,\ep_{\mft}=\sum_{\mft:\ul{L_j}\nearrow\ul{N}}A_\mft^j\,\ep_{\mft}.
	\]
	We call this an \emph{antisymmetrized} $L_j$-form. 
	
	Let $\vec{M}=(M_1,\ldots,M_J)$, and let $\vec{L}=(L_1,\ldots,L_J)$ such that each $L_j$ is even. Let $N=\vec{L}\cdot\vec{M}$, and let $\Ld=(L_1,\ldots,L_1,L_2,\ldots,L_2,\ldots,L_J,\ldots,L_J)=(\ld_1,\ld_2,\ldots,\ld_K)$ with each $L_j$ appearing $M_j$ times. Let $K=\sum_{j=1}^JM_j$. For $\sm\in S_N$, write $\sm=(\sm_1|\cdots|\sm_K)$ so that each $\sm_{k}$ is the restriction of $\sm:\ul{N}\to \ul{N}$ to a subset of size $\ld_k=L_{j}$ of which there are $M_{j}$ many. In particular, when $\sm\in{\rm{Sh}}(\Ld)$, we have $\sm_{k}:\ul{L_j}\nearrow \ul{N}$. 
	
	\begin{lem} \label{af}Under the above assumptions (particularly requiring each $L_j$ to be even) and definitions (such as how to obtain $\af_j$ from the $Q^j_{\mft}$),  we have
		\[
		\sum_{\sm\in S_N}\sgn(\sm)Q_{\sm_1}^1\cdots Q_{\sm_K}^J=\int\frac{\af_1^{\wedge_\shuffle M_1}\wedge_\shuffle\cdots\wedge_\shuffle\af_J^{\wedge_\shuffle M_J}}{K!}\,d\ep_{\rmn{vol}}.
		\]
		In particular, when $J=1$, the right hand side is ${\rm{PF}}_\shuffle(\af_1)$, where the subscript $\shuffle$ is added to emphasize the coefficients $A_\mft^j$ are in  $F\langle I\rangle_\shuffle$ in which multiplication of coefficients is done by $\shuffle$.
	\end{lem}
	
	Though this lemma holds for a more general collection of $Q_{\mft}^j$, we should think of the left hand side as being an $N\times N$ determinant. A single $Q_{\sm_k}^j$ is a product of $L_j$ many entries from the matrix. $A_{\sm_k}^j$, the antisymmetrization of $Q_{\sm_k}^j$, is the determinant of an $L_j\times L_j$ minor selected by $\sm_k$. In summary, this Lemma \ref{af} transforms any determinantal integrand into one for which Chen's Lemma will apply. This is functionally similar to the Laplace expansion of the determinant (over complimentary $L_j\times L_j$ minors) which Sinclair uses in his proofs.

	\begin{proof} Starting with the right hand side, note each factor  $\af_j$ is a sum of $A_\mft^j\ep_{\mft}$'s. If we expand the product of the sums, each summand will be a product of some $A_\mft^j\ep_{\mft}$'s. Any time $\mft:\underline{L_j}\nearrow\underline{N}$ and $\mathfrak{s}:\underline{L_k}\nearrow\underline{N}$ have overlapping ranges, $\ep_{\mft}\wedge\ep_\mathfrak{s}=0$. Thus, each nonzero summand corresponds to a permutation $(\mft_1|\ldots|\mft_K)=\sm\in S_N$. Since each $\mft_k$ is an increasing function, we have $\sm\in {\rm{Sh}}(\Ld)$. We use the Berezin integral $\int d\ep_{\rmn{vol}}$ because it sends  $\ep_\sm=\ep_{\mft_1}\wedge \cdots\wedge \ep_{\mft_K}$ to $\sgn(\sm)$ and picks up the coefficients $A_{\sm_1}^1\shuffle\cdots\shuffle A_{\sm_K}^J$. We can rewrite the right hand side as
		\[
		{\rm{RHS}}=\frac{1}{K!}\sum_{\sm\in{\rm{Sh}}(\Ld)}\sgn(\sm)A_{\sm_1}^1\shuffle\cdots\shuffle A_{\sm_K}^J.
		\]
		We eliminate the factorial denominator by requiring $\sm_{1}(1)<\sm_2(1)<\ldots<\sm_K(1)$. 
		\[
		{\rm{RHS}}=\sum_{\sm\in{\rm{Sh}}^\circ(\Ld)}\sgn(\sm)A_{\sm_1}^1\shuffle\cdots\shuffle A_{\sm_K}^J.
		\]
		Expanding each $A_{\sm_k}^{j}$ according to the definition, we get
		\[
		{\rm{RHS}}=\sum_{\sm\in{\rm{Sh}}^\circ(\Ld)}\sgn(\sm)\left(\sum_{\ta_1\in S_{L_1}}\sgn(\ta_1)Q_{\sm_1\circ \ta_1}^1\right)\shuffle\cdots\shuffle\left(\sum_{\ta_K\in S_{L_J}}\sgn(\ta_K)Q_{\sm_K\circ \ta_K}^J\right).
		\]
		Collect the $\ta_k$'s as a single element of ${\rm{H}}(\Ld)\cong (S_{L_1})^{M_1}\times\cdots\times (S_{L_J})^{M_J}$ with $\sgn(\ta)=\sgn(\ta_1)\cdots\sgn(\ta_K)$, so that 
		\[
		{\rm{RHS}}=\sum_{\sm\in{\rm{Sh}}^\circ(\Ld)}\sgn(\sm)\sum_{\ta\in {\rm{H}}(\Ld)}\sgn(\ta)Q_{\sm_1\circ\ta}^1\shuffle\cdots\shuffle Q_{\sm_K\circ\ta}^J.
		\]
		Next, we apply an identity of the $\shuffle$ operation. Note, we are shuffling individual letters together, not strings of letters, so the sum is over $\pi\in S_K={\rm{Sh}}(1,\ldots,1)$. The action on the subscripts can also be viewed as a permutation of the $K$ many ``blocks" of $N$ as prescribed by the partition $\Ld$. Thus,
		\begin{align*}
			{\rm{RHS}}&=\sum_{\sm\in{\rm{Sh}}^\circ(\Ld)}\sgn(\sm)\sum_{\ta\in {\rm{H}}(\Ld)}\sgn(\ta)\sum_{\pi\in S_{K}}Q_{\sm_{\pi(1)}\circ\ta}^1\cdots Q_{\sm_{\pi(K)}\circ\ta}^J\\&=\sum_{\ph\in S_N}\sgn(\ph)Q_{\ph_1}^1\cdots Q_{\ph_K}^J. \end{align*}
		Note, the equality in the last line is not an obvious one but follows from Lemma \ref{decomp}. In the context of Lemma \ref{af}, all of the block sizes $L_j$ are even, so $\sgn (\pi)=1$ for all $\pi \in {\rm{Bl}}(\Ld)$. Thus, $\sgn(\ph)=\sgn(\sm\circ\pi\circ \ta)=\sgn(\sm)\sgn(\ta)$. 
		
	\end{proof}

	\subsection{Confluent Vandermonde determinant}
	\label{ssec:confluent}
	
	In this subsection, we demonstrate that the joint probability density function can be written as a determinant (to which we can apply Lemma \ref{af}). Fix charge vector $\vec{L}$, population vector $\vec{M}$, and location vector $\textbf{x}$ as in \autoref{sec:setup}. Recall $N=\sum_{j=1}^JM_jL_j$. Let $\vec{f}=\left\{f_n\right\}_{n=1}^N$ be a family of $\max(L_1,\ldots,L_J)-1$ times differentiable functions. For each $j$, define the $N\times L_j$ matrix
	\[
	V^{L_j}(x)=\left[D^{l-1}f_n(x)\right]_{n,l=1}^{N,L_j}.
	\]
	For each $\textbf{x}^j\in \R^{M_j}$, define the $N\times M_jL_j$ matrix
	\[
	V^{L_j,M_j}(\textbf{x}^j)=\begin{bmatrix}
		V^{L_j}(x_1^j) & V^{L_j}(x_2^j) & \cdots & V^{L_j}(x_{M_j}^j)
	\end{bmatrix}.
	\]
	Finally, define the $N\times N$ matrix
	\[
	V^{\vec{L},\vec{M}}({\textbf{x}})=\begin{bmatrix}
		V^{L_1,M_1}(\textbf{x}^1) & V^{L_2,M_2}(\textbf{x}^2) & \cdots & V^{L_J,M_J}(\textbf{x}^J)
	\end{bmatrix},
	\]
	in which each variable $x_m^j$ appears in $L_j$ many consecutive columns, generated from $\vec{f}$ by taking derivatives. The Wronskians which appear in \autoref{sec:results} are merely the determinants of the univariate $L_j\times L_j$ minors of this matrix. 
	
	With the additional restriction that $\vec{f}$ be a complete $N$-family of monic polynomials, we call $V^{\vec{L},\vec{M}}({\textbf{x}})$ the confluent Vandermonde matrix (with respect to shape $\vec{L},\vec{M}$) in variables ${\textbf{x}}$. Under these conditions (on $\vec{f}$), it is well known \cite{Meray1899} that
	\[
	\det V^{\vec{L},\vec{M}}({\textbf{x}})=\prod_{j=1}^J\prod_{m<n}\left(x_n^j-x_m^j\right)^{ L_j^2}\times \prod_{j<k}\prod_{m=1}^{M_j}\prod_{n=1}^{M_k}\left(x_n^k-x_m^j\right)^{ L_jL_k}.
	\]
	Note, we can only construct whole numbers of columns for each variable. This is where our restrictions on $\vec{L}$ and $\bt$ come from. On the physical side, we only consider whole number charges for our particles. Using the above confluent Vandermonde determinant identity, we get
	\[
	\Om_{\vec{M}}({\textbf{x}})=W_{\vec{M}}(\textbf{x})\left|\det V^{\vec{L},\vec{M}}({\textbf{x}})\right|.
	\]
	Moreover, if we define the $N\times L_j$ matrices
	\[
	H^{L_j}(x)=\exp(-U(x))V^{L_j}(x)=\left[\exp(-U(x))\cdot D^{l-1}f_n(x)\right]_{n,l=1}^{N,L_j}
	\]
	and the combined $N\times N$ matrix $H^{\vec{L},\vec{M}}({\textbf{x}})$, then
	\[
	\left|\det H^{\vec{L},\vec{M}}({\textbf{x}})\right|=\prod_{j=1}^J\prod_{m=1}^{M_j}\exp\left(- L_jU(x_m^j)\right)\left|\det V^{\vec{L},\vec{M}}({\textbf{x}})\right|=W_{\vec{M}}(\textbf{x})\left|\det V^{\vec{L},\vec{M}}({\textbf{x}})\right|.
	\]
	Thus, integration of $V^{\vec{L},\vec{M}}({\textbf{x}})$ with respect to the $d\mu_j$'s is equivalent to integration of $H^{\vec{L},\vec{M}}({\textbf{x}})$ with respect to Lebesgue measure. For us to use Lemma \ref{af}, it is important the entire integrand $\Om_{\vec{M}}(\textbf{x})$ be determinantal, with the extra weight functions $W_{\vec{M}}(\textbf{x})$ incorporated into $H^{\vec{L},\vec{M}}({\textbf{x}})$.
	
	As an example, consider one charge 2 particle, one charge 3 particle, and three charge 1 particles with potential $U(x)=x^2$. This gives us $\vec{L}=(2,3,1)$, $\vec{M}=(1,1,3)$, and $N=8$. For simplicity, we will use the variables $\textbf{x}=(a,b,c_1,c_2,c_3)$. Let $\vec{f}=\{x^{n-1}\}_{n=1}^N$. Then the three columns corresponding to the charge 3 particle are
	\[
	V^3(b)=\begin{bmatrix}
		1 & 0 & 0 \\
		b & 1 & 0 \\
		b^2 & 2b & 1 \\
		b^3 & 3b^2 & 3b \\
		\vdots &  & \vdots \\
		b^7 & 7b^6 & 21b^5
	\end{bmatrix}.
	\]
	In the third column, we have not just the second derivative but also a denominator of $2!$. One consequence of these $l!$ denominators in $D^{l-1}$ is that we get 1's on this diagonal. We also get three columns corresponding to the three charge 1 particles
	\[
	V^{1,3}(c_1,c_2,c_3)=\begin{bmatrix}
		1 & 1 & 1 \\
		c_1 & c_2 & c_3 \\
		c_1^2 & c_2^2 & c_3^2 \\
		\vdots &  & \vdots \\
		c_1^7 & c_2^7 & c_3^7
	\end{bmatrix}.
	\]
	Together, the full $8\times 8$ confluent Vandermonde matrix is
	\[
	V^{\vec{L},\vec{M}}(\textbf{x})=\begin{bmatrix}
		1 & 0 & 1 & 0 & 0 & 1 & 1 & 1\\
		a & 1 & b & 1 & 0 & c_1 & c_2 & c_3\\
		a^2 & 2a & b^2 & 2b & 1 & c_1^2 & c_2^2 & c_3^2 \\
		a^3 & 3a^2 & b^3 & 3b^2 & 3b & c_1^3 & c_2^3 & c_3^3 \\
		\vdots & & & & & & & \vdots \\
		a^7 & 7a^6 & b^7 & 7b^6 & 21b^5 & c_1^7 & c_2^7 & c_3^7 
	\end{bmatrix}.
	\]
	We obtain $H^{\vec{L},\vec{M}}(\textbf{x})$ by multiplying the first two columns by $\exp(-a^2)$, the next three columns by $\exp(-b^2)$, and the last three columns by the appropriate $\exp(-c_j^2)$. This changes the determinant by
	\[
	W_{\vec{M}}(\textbf{x})=\exp(-2a^2)\exp(-3b^2)\exp(-c_1^2)\exp(-c_2^2)\exp(-c_3^2).
	\]

	\subsection{Absolute value of determinants}
	\label{ssec:absdet}
	
	Though our Boltzmann factor integrand $\Om_{\vec{M}}(\textbf{x})$ is recognizably determinantal, we still need to remove the absolute value before we can apply Lemma \ref{af} (and later Lemma \ref{chen}). This can be done by decomposing the domain of integration into subsets over which the sign of the determinant is constant. Namely, we use totally ordered subsets $\Dt_N(\sm)$ over which the differences in the confluent Vandermonde determinant never change signs. These smaller domains of integration are exactly the ones which allow us to apply (Chen's) Lemma \ref{chen}.
	
	As in \autoref{sec:results}, suppose $L_j$ is even for $1\leq j\leq r$. Let $ K_e=\sum_{j=1}^rM_j$ be the total number of particles with even charge, and let $ K_o=\sum_{j=r+1}^JM_j$ be the total number of particles with odd charge. Relabel
	\[
	y_1=x_1^{1}, \hspace{10mm} y_2=x_2^{1}, \hspace{10mm} \cdots \hspace{10mm} y_{M_1}=x_{M_1}^{1}, \hspace{10mm} y_{M_1+1}=x_{1}^{2} \hspace{10mm}\cdots \hspace{10mm} y_{K_e}=x_{M_r}^{r}
	\]
	so that $\vec{y}=(\textbf{x}^{1}, \textbf{x}^{2}, \ldots, \textbf{x}^{r})$ gives the locations of the particles of even charge. Similarly, relabel 
	\[
	w_1=x_1^{r+1},\hspace{10mm} \cdots \hspace{10mm} w_{M_{r+1}}=x_{M_{r+1}}^{r+1}, \hspace{10mm} w_{M_{r+1}+1}=x_{1}^{r+2} \hspace{10mm}\cdots \hspace{10mm} w_{K_e}=x_{M_J}^{J}
	\]
	so that $\vec{w}=(\textbf{x}^{r+1}, \textbf{x}^{r+2}, \ldots, \textbf{x}^{J})$ gives the locations of the particles of odd charge. Define $\ld_j^e$ to be the $L_k$ which corresponds to $y_j$ so that
	\[
	\Ld^e=(\ld_1^e,\ldots,\ld_{K_e}^e)=(L_{1},\ldots,L_{1},L_{2},\ldots,L_{2},\ldots,L_{r},\ldots,L_r)
	\]
	gives the list of even charges, with each $L_j$ appearing $M_j$ times. Similarly define $\ld_j^o$ for corresponding $w_j$ so that
	\[
	\Ld^o=(\ld_1^o,\ldots,\ld_{K_o}^o)=(L_{r+1},\ldots,L_{r+1},L_{r+2},\ldots,L_{r+2},\ldots,L_{J},\ldots,L_J)
	\]
	gives the list of the odd charges. We will treat all charges $\ld_j^e$ and $\ld_j^o$ as distinct until it is relevant to recall which charges are repeated (and how many times each).
	
	For each $(\sm,\ta)\in S_{K_e}\times S_{K_o}$, define $V_{\sm\ta}^{\vec{L},\vec{M}}({\textbf{x}})$ to be the matrix
	\[
	V^{\vec{L},\vec{M}}_{\sm,\ta}({\textbf{x}})=\begin{bmatrix}
		V^{\ld^e_{\sm^{-1}(1)}}(y_{\sm^{-1}(1)}) & \cdots & V^{\ld^e_{\sm^{-1}(K_e)}}(y_{\sm^{-1}(K_e)}) & V^{\ld^o_{\ta^{-1}(1)}}(w_{\ta^{-1}(1)}) & \cdots & V^{\ld^o_{\ta^{-1}(K_o)}}(w_{\ta^{-1}(K_o)})
	\end{bmatrix},
	\]
	obtained from $V^{\vec{L},\vec{M}}({\textbf{x}})$ by permuting the columns so that the columns with $y_{\sm^{-1}(1)}$ come first (of which there are $\ld^e_{\sm^{-1}(1)}$ many), then all of the columns with $y_{\sm^{-1}(2)}$ come next (of which there are $\ld^e_{\sm^{-1}(2)}$ many), and so on until the $y_j$ are exhausted, doing the same for the $w_j$.

	Using the same confluent Vandermonde determinant identity from \autoref{ssec:confluent}, we get
	\begin{align*}
		\det V^{\vec{L},\vec{M}}_{\sm,\ta}({\textbf{x}})&=\prod_{j<k}(y_{\sm^{-1}(k)}-y_{\sm^{-1}(j)})^{\ld^e_{\sm^{-1}(k)}\ld^e_{\sm^{-1}(j)}}\times \prod_{j=1}^{K_e}\prod_{k=1}^{K_o}(w_{\sm^{-1}(k)}-y_{\ta^{-1}(j)})^{\ld^o_{\sm^{-1}(k)}\ld^e_{\ta^{-1}(j)}}\\&\hspace{10mm}\times \prod_{j<k}(w_{\ta^{-1}(k)}-w_{\ta^{-1}(j)})^{\ld^o_{\ta^{-1}(k)}\ld^o_{\ta^{-1}(j)}}.
	\end{align*}
	Next, consider ${\textbf{x}}=(\vec{y},\vec{w})\in \Dt_{K_e}(\sm)\times \Dt_{K_o}(\ta)$ in which the even charged particles (located by $\vec{y}$) are ordered according to $\sm$ and the odd charged particles (located by $\vec{w}$) are ordered according to $\ta$. In particular, $w_{\ta^{-1}(k)}>w_{\ta^{-1}(j)}$ whenever $j<k$. Thus, all differences in the third product are positive. Additionally, each difference in the first and second products have even exponents $\ld^e_j$. Thus, 
	\[
	\left|\det V^{\vec{L},\vec{M}}_{\sm,\ta}({\textbf{x}})\right|=\det V^{\vec{L},\vec{M}}_{\sm,\ta}({\textbf{x}})
	\]
	on the domain $\Dt_{K_e}(\sm)\times \Dt_{K_o}(\ta)$. 
	
	Note, permuting the variables $\vec{y}$ involves permuting blocks of even numbers of columns at a time, leaving the determinant of $V^{\vec{L},\vec{M}}({\textbf{x}})$ unchanged. In contrast, permuting variables $\vec{w}$ involves permuting blocks of odd numbers of columns at a time, changing the determinant by $\sgn(\ta)$. Thus, 
	\[
	\left|\det V^{\vec{L},\vec{M}}({\textbf{x}})\right|=\left|\sgn(\ta)\det V^{\vec{L},\vec{M}}_{\sm,\ta}({\textbf{x}})\right|=\det V^{\vec{L},\vec{M}}_{\sm,\ta}({\textbf{x}})
	\]
	on the domain $\Dt_{K_e}(\sm)\times \Dt_{K_o}(\ta)$. Analogous results hold if we replace $V^{\vec{L},\vec{M}}({\textbf{x}})$ with $H^{\vec{L},\vec{M}}({\textbf{x}})$ (as defined in \autoref{ssec:confluent}).
	
	Recall the example from \autoref{ssec:confluent} in which we have a single particle of even charge 2. Then $K_e=1$, $\vec{y}=(a)$,  and $\Ld^e=(2)$. There is one particle of odd charge 3, and there are three particles of odd charge 1. Then $K_o=4$, $\vec{w}=(b,c_1,c_2,c_3)$, and $\Ld^o=(3,1,1,1)$. Let $\ta$ be the permutation which swaps $b$ with $c_1$. The new matrix (with permuted columns) is given by
	\[
	V^{\vec{L},\vec{M}}_{\id,\ta}(\textbf{x})=\begin{bmatrix}
		1 & 0 & 1 & 1 & 0 & 0  & 1 & 1\\
		a & 1 & c_1 & b & 1 & 0  & c_2 & c_3\\
		a^2 & 2a & c_1^2 & b^2 & 2b & 1  & c_2^2 & c_3^2 \\
		a^3 & 3a^2 & c_1^3 & b^3 & 3b^2 & 3b  & c_2^3 & c_3^3 \\
		\vdots & & & & & & & \vdots \\
		a^7 & 7a^6 & c_1^7 & b^7 & 7b^6 & 21b^5  & c_2^7 & c_3^7 
	\end{bmatrix}
	\]
	Note, swapping two variables requires more than just swapping two columns. We swap the entire three-column block $V^3(b)$ with the one-column block $V^1(c_1)$. 
	
	\section{Canonical ensemble}
	\label{sec:canonical}
	
	With the modification to the integrand outlined in \autoref{ssec:absdet}, we can now decompose $Z_{\vec{M}}$ into integrals without absolute value, provided we divide the domain of integration appropriately. Explicitly, 
	\begin{align*}
		Z_{\vec{M}}&=\frac{1}{M_1!M_2!\cdots M_J!}\int_{\R^{K_e}}\int_{\R^{K_o}}\Om_{\vec{M}}({\textbf{x}})\,dy_1\cdots dy_{K_e}\,dw_{1}\cdots dw_{K_o}\\&=\frac{1}{M_1!M_2!\cdots M_J!}\sum_{\sm\in S_{K_e}}\sum_{\ta\in S_{K_o}}\int_{\Dt_{K_e}(\sm)}\int_{\Dt_{K_o}(\ta)}\left|\det V^{\vec{L},\vec{M}}({\textbf{x}})\right|W_{\vec{M}}({\textbf{x}})\,dy_1\cdots dy_{K_e}\,dw_{1}\cdots dw_{K_o}
		\\&=\frac{1}{M_1!M_2!\cdots M_J!}\sum_{\sm\in S_{K_e}}\sum_{\ta\in S_{K_o}}\int_{\Dt_{K_e}(\sm)}\int_{\Dt_{K_o}(\ta)}\det H^{\vec{L},\vec{M}}_{\sm,\ta}({\textbf{x}})\,dy_1\cdots dy_{K_e}\,dw_{1}\cdots dw_{K_o},
	\end{align*}
	summing over all totally ordered subsets $\Dt_{K_e}(\sm)\subset \R^{K_e}$ and $\Dt_{K_o}(\ta)\subset \R^{K_o}$.

	\subsection{When all $L_j$ are even}
	\label{ssec:alleven}
	
	For $\pi=(\pi_1|\ldots|\pi_{K_e})\in S_N$ and complete $N$-family of monic polynomials $\vec{p}=\{p_n\}_{n=1}^N$, define 
	\[
	Q_{\pi_j}^k(x)=\prod_{l=1}^{\ld^e_{k}}\exp(-U(x))\cdot D^{l-1}p_{\pi_j(l)}(x),
	\]
	so that we can write
	\[
	\det H^{\vec{L},\vec{M}}_\sm(\vec{y})=\sum_{\pi\in S_{N}}\sgn(\pi)Q_{\pi_1}^{\sm^{-1}(1)}(y_{\sm^{-1}(1)})\cdots  Q_{\pi_{K_e}}^{\sm^{-1}(K_e)}(y_{\sm^{-1}(K_e)}).
	\]
	Note, each $Q^k_{\pi_j}$ is a product of $\ld_k^e$ many matrix entries which we have grouped together. These entries are taken from the $\ld_k^e$ many derivative columns for the one variable $y_k$. Writing the determinant this way allows to invoke Lemma \ref{af} which gives us
	\[
	\det H^{\vec{L},\vec{M}}_\sm(\vec{y})=\int\frac{\af_{\sm^{-1}(1)}\wedge_\shuffle\cdots\wedge_\shuffle \af_{\sm^{-1}(K_e)}}{K_e!}\,d\ep_{\rmn{vol}},
	\]
	where each $\af_k$ is defined by
	\[
	\af_k(x)=\sum_{\mft:\ul{\ld^e_k}\nearrow\ul{N}}A_\mft^k(x)\,\,\ep_{\mft}=\sum_{\mft:\ul{\ld^e_k}\nearrow\ul{N}}\sum_{\ta\in S_{\ld^e_k}}\sgn(\ta)Q_{\mft\circ \ta}^k(x)\,\,\ep_{\mft}=\sum_{\mft:\ul{\ld^e_k}\nearrow\ul{N}}\exp(-\ld^e_kU(x)){\rm{Wr}}(\vec{p}_\mft,x)\,\,\ep_{\mft}.
	\]
	As mentioned in \autoref{ssec:confluent}, these Wronskians are the determinants of the univariate $L_j\times L_j$ minors of the confluent Vandermonde matrix. Recall from \autoref{sec:results},
	\[
	\gm_j=\sum_{\mft:\ul{L_j}\nearrow\ul{N}}\int_\R {\rm{Wr}}(\vec{p}_\mft,x)\,d\mu_j(x)\,\ep_{\mft},
	\]
	with $d\mu_j(x)=\exp(-L_jU(x))\,dx$. Starting from 
	\[
	Z_{\vec{M}}=\frac{1}{M_1!M_2!\cdots M_J!}\sum_{\sm\in S_{K_e}}\int_{\Dt_{K_e}(\sm)}\det H^{\vec{L},\vec{M}}_{\sm}(\vec{y})\,dy_1\cdots dy_{K_e},
	\]
	relabeling the variables $x_j=y_{\sm^{-1}(j)}$ produces
	\[
	Z_{\vec{M}}=\frac{1}{M_1!M_2!\cdots M_J!}\sum_{\sm\in S_{K_e}}\int_{\Dt_{K_e}(\id)}\det H^{\vec{L},\vec{M}}_{\sm}(\vec{x})\,dx_1\cdots dx_{K_e}.
	\]
	By Lemma \ref{af}, the determinant can be rewritten in the exterior shuffle algebra. Thus, 
	\[
	Z_{\vec{M}}=\frac{1}{M_1!M_2!\cdots M_J!}\sum_{\sm\in S_{K_e}}\left\langle \int\frac{\af_{\sm^{-1}(1)}\wedge_\shuffle\cdots\wedge_\shuffle \af_{\sm^{-1}(K_e)}}{K_e!}\,d\ep_{\rmn{vol}} \right\rangle.
	\]
	Applying (Chen's) Lemma \ref{chen} sends integration of shuffle products to ordinary products of integrals. Thus,
	\[
	Z_{\vec{M}}=\frac{1}{M_1!M_2!\cdots M_J!}\sum_{\sm\in S_{K_e}}\int\frac{\gm_1^{\wedge M_1}\wedge\cdots\wedge\gm_J^{\wedge M_J}}{K_e!}\,d\ep_{\rmn{vol}}.
	\]
	Note, there exist $M_j$ many $k$ for which $\ld_k^e=L_j$, so each factor $\gm_j$ appears $M_j$ times. This happens independent of $\sm\in S_{K_e}$, of which there are $|S_{K_e}|=K_e!$ many. Thus,
	\[
	Z_{\vec{M}}=\frac{1}{M_1!M_2!\cdots M_J!}\int\gm_1^{\wedge M_1}\wedge\cdots\wedge\gm_J^{\wedge M_J}\,d\ep_{\rmn{vol}}.
	\]
	We have now proven the following lemma:
	
	\begin{lem} \label{alleven}If all $L_j$ are even, then
		\[
		Z_{\vec{M}}=\int\frac{\gm_1^{\wedge M_1}}{M_1!}\wedge\cdots\wedge\frac{\gm_J^{\wedge M_J}}{M_J!}\,d\ep_{\rmn{vol}}.
		\]
	\end{lem}
	
	\subsection{Even number of odd charges}
	\label{ssec:evenodd}
	
	In this subsection, we assume all $L_j$ are odd, but the total number of particles $K_o=\sum_{j=1}^JM_j=2K$ is even. This happens, for example, when total charge $N$ is even. Recall from \autoref{sec:results}, 
	\[
	\et_{j,k}=\sum_{\mft:\underline{L_j}\nearrow\underline{N}}\sum_{\mathfrak{s}:\underline{L_k}\nearrow\underline{N}}\int\int_{x<y}{\rm{Wr}}(\vec{p}_{\mft},x){\rm{Wr}}(\vec{p}_{\mathfrak{s}},y)\,d\mu_j(x)d\mu_k(y)\,\ep_{\mft}\wedge\ep_{\mathfrak{s}}.
	\]
	Following the same argument as in \autoref{ssec:alleven}, we can apply our Lemma \ref{af} provided there exists antisymmetrized \emph{even} forms $\af_{j,k}$ such that each $\et_{j,k}$ is obtained from $\af_{j,k}$ by applying our functional $\langle\cdot \rangle$ to the coefficients of $\af_{j,k}$. To this end, define $\af_{j,k}$ by
	\[
	\af_{j,k}=\sum_{\mft:\ul{\ld_j^o+\ld^o_k}\nearrow \ul{N}}A_{\mft}^{j,k}\ep_\mft=\sum_{\mft:\ul{\ld_j^o+\ld^o_k}\nearrow \ul{N}}\,\sum_{\pi\in S_{\ld_j^o+\ld^o_k}}\sgn(\pi)Q_{\mft\circ \pi}^{j,k}\,\ep_\mft,
	\]
	where
	\[
	Q_{\mft\circ\pi}^{j,k}(x,y)=\prod_{l=1}^{\ld_j^o}\exp(-U(x))\cdot D^{l-1}p_{\mft\circ \pi(l)}(x)\prod_{l=1}^{\ld^o_k}\exp(-U(y))\cdot D^{l-1}p_{\mft\circ \pi(\ld_j^o+l)}(y),
	\]
	then $\af_{j,k}$ is an antisymmetrized $(\ld_j^o+\ld^o_k)$-form by construction. As before, we should think of $Q_{\mft}^{j,k}$ as taking one entry from each of the $\ld^o_j$ many derivative columns for one variable $x$ and then one entry from each of the $\ld^o_k$ many derivative columns for the next variable $y$. Pairing an odd number of entries with another odd number of entries produces an even form. It remains to be shown that the antisymmetrizations $A_\mft^{j,k}$ are composed of complementary Wronskian minors. 
	
	For $\mft:\ul{\ld_j^o+\ld^o_k}\nearrow \ul{N}$, let $H_{\ta,\mft}^{\vec{L},\vec{M}}(w_j,w_k)$ denote the $(\ld_j^o+\ld^o_k)\times (\ld_j^o+\ld^o_k)$ minor of $H_{\ta}^{\vec{L},\vec{M}}(\vec{w})$ comprised of rows $\mft(1),\ldots,\mft(\ld_j^o+\ld^o_k)$ taken from the $\ld_j^o$ columns in $w_j$ and the $\ld^o_k$ columns in $w_k$. When viewed as a two variable function, $A_{\mft}^{j,k}$ is the determinant of this minor. Explicitly, 
	\[
	A_{\mft}^{j,k}(w_j,w_k)=\det H_{\ta,\mft}^{\vec{L},\vec{M}}(w_j,w_k).
	\]
	For $\mft_1:\ul{\ld_j^o}\nearrow \ul{\ld_j^o+\ld^o_k}$, let $H_{\ta,\mft_1}^{\vec{L},\vec{M}}(w_j)$ denote the $\ld_j^o\times \ld_j^o$ minor of $H_{\ta,\mft}^{\vec{L},\vec{M}}(w_j,w_k)$ comprised of rows $\mft_1(1),\ldots,\mft_1(\ld_j^o)$ taken from the $\ld_j^o$ columns in $w_j$. Similarly, for $\mft_2:\ul{\ld^o_k}\nearrow \ul{\ld_j^o+\ld^o_k}$, let $H_{\ta,\mft_2}^{\vec{L},\vec{M}}(w_k)$ denote a $\ld^o_k\times \ld^o_k$ minor in $w_k$. By the Laplace expansion of the determinant,
	\begin{align*}
		\det H_{\ta,\mft}^{\vec{L},\vec{M}}(w_j,w_k)\,\ep_\mft&=\sum_{\mft_1:\ul{\ld_j^o}\nearrow \ul{\ld_j^o+\ld^o_k}}\sum_{\mft_2:\ul{\ld^o_k}\nearrow \ul{\ld_j^o+\ld^o_k}}\det H_{\ta,\mft_1}^{\vec{L},\vec{M}}(w_j) \det H_{\ta,\mft_2}^{\vec{L},\vec{M}}(w_k)\cdot\sgn(\mft_1,\mft_2)\,\ep_\mft\\&=\sum_{\mft_1:\ul{\ld_j^o}\nearrow \ul{\ld_j^o+\ld^o_k}}\sum_{\mft_2:\ul{\ld^o_k}\nearrow \ul{\ld_j^o+\ld^o_k}}\det H_{\ta,\mft_1}^{\vec{L},\vec{M}}(w_j) \det H_{\ta,\mft_2}^{\vec{L},\vec{M}}(w_k)\cdot \ep_{\mft_1}\wedge \ep_{\mft_2}\\&=\sum_{\mft_1:\ul{\ld_j^o}\nearrow \ul{\ld_j^o+\ld^o_k}}\sum_{\mft_2:\ul{\ld^o_k}\nearrow \ul{\ld_j^o+\ld^o_k}}\exp(-\ld^o_jU(w_j)){\rm{Wr}}(\vec{p}_{\mft_1},w_j)\exp(-\ld^o_kU(w_k)){\rm{Wr}}(\vec{p}_{\mft_2},w_k)\cdot \ep_{\mft_1}\wedge \ep_{\mft_2}.
	\end{align*}
	Thus, $\et_{j,k}$ is the result of applying $\langle\cdot\rangle_2$ to the two-variable coefficients of $\af_{j,k}$ as desired. 
	
	Proceeding as we did in \autoref{ssec:alleven} (applying Lemma \ref{af} in the third line below and Lemma \ref{chen} in the last line), 
	\begin{align*}
		Z_{\vec{M}}&=\frac{1}{M_1!M_2!\cdots M_J!}\sum_{\ta\in S_{K_o}}\int_{\Dt_{K_o}(\ta)}\det H^{\vec{L},\vec{M}}_{\ta}(\vec{w})\,dw_1\cdots dw_{K_o}\\&=\frac{1}{M_1!M_2!\cdots M_J!}\sum_{\ta\in S_{K_o}}\int_{\Dt_{K_o}(\id)}\det H^{\vec{L},\vec{M}}_{\ta}(\vec{x})\,dx_1\cdots dx_{K_e}\\&=\frac{1}{M_1!M_2!\cdots M_J!}\sum_{\ta\in S_{K_o}}\left\langle \int\frac{\af_{\ta^{-1}(1),\ta^{-1}(2)}\wedge_\shuffle\cdots\wedge_\shuffle \af_{\ta^{-1}(K_o-1),\ta^{-1}(K_o)}}{K!}\,d\ep_{\rmn{vol}} \right\rangle\\&=\frac{1}{M_1!M_2!\cdots M_J!}\sum_{\ta\in S_{K_o}}\int\frac{1}{K!}\bigwedge_{j=1}^J\bigwedge_{k=1}^J\et_{j,k}^{\wedge M_{\ta}^{j,k}}\,d\ep_{\rmn{vol}},
	\end{align*}
	where $M_{\ta}^{j,k}$ is the number of times $\ld^o_{\ta^{-1}(2n-1)}=L_j$ while $\ld^o_{\ta^{-1}(2n)}=L_k$, and $K=K_o/2=\sum_{j,k}M_\ta^{j,k}$ is the total number of factors in the wedge product. Note, these $M_\ta^{j,k}$ exponents depend on $\ta$, so we're not able to drop the sum. However, the $M_j$ many particles with the same charge $L_j$ are indistinguishable. Restricting to shuffle permutations removes the redundancy in permuting variables which have the same $L_j$. We have now proven another lemma:
	
	\begin{lem} \label{evenodd} If all $L_j$ are odd, but total charge $N$ is even, then
		\[
		Z_{\vec{M}}=\sum_{\ta\in {\rm{Sh}}(M_1,\ldots,M_J)}\int\frac{1}{K!}\bigwedge_{j=1}^J\bigwedge_{k=1}^J\et_{j,k}^{\wedge M_{\ta}^{j,k}}\,d\ep_{\rmn{vol}}.
		\]
	\end{lem}
	Recall the example from \autoref{ssec:confluent} and \autoref{ssec:absdet}. Modify this example by replacing the even charge 2 particle with an odd charge 3 particle, leaving the other particle of odd charge 3 and three particles of odd charge 1. The three columns in variables $a$ and $b$ produce $3\times 3$ Wronskian minors, while the remaining columns in the variables $c_1,c_2,c_3$ produce $1\times 1$ Wronskian minors. 
	
	Under the identity permutation, we pair the three columns in variable $a$ with the three columns in variable $b$ to produce $\et_{3,3}$. Pairing the one column in variable $c_1$ with the one column in variable $c_2$ produces $\et_{1,1}$. 
	
	Under the permutation $\ta$ which previously swapped the three columns in variable $b$ with the one column in variable $c_1$, we pair variable $a$ (charge 3) with $c_1$ (charge 1), and we pair $b$ (charge 3) with $c_2$ (charge 1). After integrating out all the variables, the result is two copies of $\et_{3,1}$.
	
	The permutation which swaps $a$ with $b$ produces the same $\et_{3,3}$ as the identity permutation. To avoid this redundancy, we consider only shuffle permutations. The permutation which moves $c_1$ to the front (ordering the variables as $c_1, a, b, c_2, c_3$) produces $\et_{1,3}$ followed by the distinct $\et_{3,1}$. 
	
	Note, in this example as stated, the last variable $c_3$ is unpaired because we have an odd number of variables. In the next subsection, we demonstrate how to fix this. Amending an extra column to the confluent Vandermonde matrix allows us to pair the last variable with a placeholder. Once integrated, this last ``pair" produces the single Wronskian form $\gm_j$ instead of the double Wronskian form $\et_{j,k}$.

	\subsection{Odd number of odd charges}

	In this subsection, we again assume all $L_j$ are odd, but the total number of particles $K_o=\sum_{j=1}^JM_j=2K-1$ is odd. This happens, for example, when total charge $N$ is odd. Start by modifying $H_{\ta}^{\vec{L},\vec{M}}(\vec{w})$ as follows:
	\[
	H_{\ta}^{\vec{L},\vec{M},1}(\vec{w})=\begin{bmatrix}
		H_{\ta}^{\vec{L},\vec{M}}(\vec{w}) & 0 \\ 0 & 1
	\end{bmatrix}
	\]
	In the previous subsection, we constructed even forms $\af_{j,k}$ and subsequent $\et_{j,k}$ by pairing an $L_j$ charge with an adjacent $L_k$ charge. We also showed taking determinants of appropriate minors produces an antisymmetrized form $\af_{j,k}$ (for which our Lemma \ref{af} will apply). We will do this pairing again for the first $2K-2$ variables, which makes $K-1$ pairs. Explicitly, simply define $Q_\mft^{j,k}, A_\mft^{j,k}, \af_{j,k},$ and $\et_{j,k}$ as before.

	For convenience, let $L=\ld^o_{\ta^{-1}(K_o)}$ be the charge corresponding to the last variable $w_{\ta^{-1}(K_o)}$. Construct the final $\af_{\ta^{-1}(K_o)}$ by taking $(L+1)\times (L+1)$ minors from the last $L+1$ columns of $H_{\ta}^{\vec{L},\vec{M},1}(\vec{w})$. These minors have non-zero determinant only when the last row is chosen. Thus, valid minors are entirely determined by a choice of only $L$ many other rows, and
	\[
	\af_{\ta^{-1}(K_o)}=\sum_{\mft:\ul{L}\nearrow \ul{N}}\sum_{\pi\in S_L}\sgn(\pi)Q_{\mft\circ \pi}^{\ta^{-1}(K_o)}\,\ep_{\mft}\wedge \ep_{N+1}
	\]
	is an antisymmetrized $(L+1)$-form when
	\[
	Q_{\mft\circ \pi}^{\ta^{-1}(K_o)}(x)=\prod_{l=1}^{L}\exp(-U(x))\cdot D^{l-1}p_{\pi(l)}(x).
	\]
	As in \autoref{ssec:alleven}, applying $\langle\cdot\rangle_1$ to $\af_{\ta^{-1}(K_o)}$ produces $\gm_{\ta^{-1}(K_o)}\wedge \ep_{N+1}$. Then the expression for $Z_{\vec{M}}$ changes only slightly to include this new (unpaired, single Wronskian) factor:
	
	\begin{lem} \label{oddodd} If all $L_j$ are odd, and the total charge $N$ is odd, then
		\[
		Z_{\vec{M}}=\sum_{\ta\in {\rm{Sh}}(M_1,\ldots,M_J)}\int\frac{1}{K!}\bigwedge_{j=1}^J\bigwedge_{k=1}^J\et_{j,k}^{\wedge M_{\ta}^{j,k}}\wedge \gm_{\ta^{-1}(K_o)} \,d\ep_{\rmn{vol}},
		\]
		where $K=(K_o+1)/2=1+\sum_{j,k}M_\ta^{j,k}$ is the total number of factors in the wedge product.
	\end{lem}
	
	Recall the example at the end of the previous subsection with two particles of charge 3 and three particles of charge 1. Under the identity permutation, we pair a charge 3 with a charge 3, pair a charge 1 with a charge 1, and leave a charge 1 unpaired. This produces $\et_{3,3}\wedge \et_{1,1}\wedge \gm_1$. Under the permutation $\ta$ which swapped the second charge 3 with the first charge 1, we got $\et_{3,1}\wedge\et_{3,1}\wedge \gm_1$. 
	
	Consider instead the permutation which puts all of the charge 1 particles before the charge 3 particles. We pair a charge 1 with a charge 1, pair the last charge 1 with a charge 3, and leave a charge 3 unpaired. This produces $\et_{1,1}\wedge \et_{1,3}\wedge \gm_3$.

	\subsection{Arbitrary charge vector}
	\label{ssec:arbitrary}
	
	Finally, we allow any mix of odd and even charges. Recall (from the beginning of \autoref{sec:canonical}),
	\[
	Z_{\vec{M}}=\frac{1}{M_1!M_2!\cdots M_J!}\sum_{\sm\in S_{K_e}}\sum_{\ta\in S_{K_o}}\int_{\Dt_{K_e}(\sm)}\int_{\Dt_{K_o}(\ta)}\det H^{\vec{L},\vec{M}}_{\sm,\ta}({\textbf{x}})\,dy_1\cdots dy_{K_e}\,dw_{1}\cdots dw_{K_o}.
	\]
	Let $N_e=\sum_{j=1}^{r}L_jM_j$ be the total charge of the even charges, and let $N_o=\sum_{j=r+1}^JL_jM_j$ be the total charge of the odd charges. By the Laplace expansion of the determinant, 
	\[
	\det H^{\vec{L},\vec{M}}_{\sm,\ta}(\vec{y},\vec{w})\,\ep_{\rmn{vol}}=\sum_{\mft:\ul{N_e}\nearrow\ul{N}}\sum_{\mathfrak{s}:\ul{N_o}\nearrow\ul{N}}\det H^{\vec{L},\vec{M}}_{\sm,\mft}(\vec{y})\det H^{\vec{L},\vec{M}}_{\ta,\mathfrak{s}}(\vec{w})\cdot\ep_\mft\wedge \ep_{\mathfrak{s}},
	\]
	where $\det H^{\vec{L},\vec{M}}_{\sm,\mft}(\vec{y})$ is an $N_e\times N_e$ minor taken only from columns in the (even charge) variables $\vec{y}$, and $\det H^{\vec{L},\vec{M}}_{\sm,\mathfrak{s}}(\vec{w})$ is an $N_o\times N_o$ minor taken only from columns in the (odd charge) variables $\vec{w}$. Note, $\ep_{\mft}\wedge \ep_\mathfrak{s}=0$ whenever these minors are not complimentary. 
	
	With the variables separated in this way, we can apply Lemma \ref{alleven} to the determinant in the even charges, and we can apply either Lemma \ref{evenodd} or Lemma \ref{oddodd} to the determinant in the odd charges. For the even charges, we have
	
	\[
	\frac{1}{M_1!\cdots M_r!}\sum_{\sm\in S_{K_e}}\int_{\Dt_{K_e}(\sm)}\det H^{\vec{L},\vec{M}}_{\sm,\mft}(\vec{y})\,dy_1\cdots dy_{K_e}\cdot\ep_\mft=\frac{{\gm^{\mft}_{1}}^{\wedge M_1}}{M_1!}\wedge \cdots \wedge \frac{{\gm^{\mft}_{r}}^{\wedge M_r}}{M_r!},
	\]
	where $\gm_j^\mft$ is subtly different from $\gm_j$ because $H^{\vec{L},\vec{M}}_{\sm,\mft}(\vec{y})$ is already an $N_e\times N_e$ minor chosen by $\mft$. Explicitly,
	\[
	\gm_j^\mft=\sum_{\mft_j:\ul{L_j}\nearrow \ul{N_e}}\int_{\R}{\rm{Wr}}(\vec{p}_{\mft\circ \mft_j},x)\,d\mu_j(x)\,\ep_{\mft_j}.
	\]
	Taking the sum over all $\mft:\ul{N_e}\nearrow \ul{N}$ gives us back the original $\gm_j$ forms
	\[
	\sum_{\mft:\ul{N_e}\nearrow \ul{N}}\frac{{\gm^{\mft}_{1}}^{\wedge M_1}}{M_1!}\wedge \cdots \wedge \frac{{\gm^{\mft}_{r}}^{\wedge M_r}}{M_r!}=\frac{\gm_{1}^{\wedge M_1}}{M_1!}\wedge \cdots \wedge \frac{\gm_{r}^{\wedge M_r}}{M_r!}.
	\]
	It is straightforward to check an analogous result holds for the determinant in the odd charges. The following lemma supersedes Lemmas \ref{alleven}, \ref{evenodd}, and \ref{oddodd}:
	
	\begin{lem} \label{supercede} Suppose $L_j$ is even for $1\leq j\leq r$, then when $N$ is even,
		\[
		Z_{\vec{M}}=\int\frac{\gm_{1}^{\wedge M_1}}{M_1!}\wedge \cdots \wedge \frac{\gm_{r}^{\wedge M_r}}{M_r!}\wedge \sum_{\ta\in {\rm{Sh}}(M_{r+1},\ldots,M_J)}\frac{1}{K!}\bigwedge_{j=1}^J\bigwedge_{k=1}^J\et_{j,k}^{\wedge M_{\ta}^{j,k}}\,d\ep_{\rmn{vol}},
		\]
		and when $N$ is odd,
		\[
		Z_{\vec{M}}=\int\frac{\gm_{1}^{\wedge M_1}}{M_1!}\wedge \cdots \wedge \frac{\gm_{r}^{\wedge M_r}}{M_r!}\wedge \sum_{\ta\in {\rm{Sh}}(M_{r+1},\ldots,M_J)}\frac{1}{K!}\bigwedge_{j=1}^J\bigwedge_{k=1}^J\et_{j,k}^{\wedge M_{\ta}^{j,k}}\wedge \gm_{\ta^{-1}(K_o)}\,d\ep_{\rmn{vol}}.
		\]
	\end{lem}

	\section{Isocharge grand canonical ensemble}
	\label{sec:isocharge}

	Recall from \autoref{sec:setup}, we want to compute
	\[
	Z_N=\sum_{\vec{L}\cdot \vec{M}=N}z_1^{M_1}z_2^{M_2}\cdots z_{J}^{M_J}Z_{\vec{M}}.
	\]

	\subsection{When all $L_j$ are even}
	\label{ssec:isoalleven}
	
	Starting from Lemma \ref{alleven}, we have
	\[
	Z_N=\sum_{\vec{L}\cdot\vec{M}=N}z_1^{M_1}\cdots z_J^{M_J}\int\frac{\gm_1^{\wedge M_1}}{M_1!}\wedge\cdots\wedge\frac{\gm_J^{\wedge M_J}}{M_J!}\,d\ep_{\rmn{vol}}.
	\]
	Recall from \autoref{ssec:berezin}, the Berezin integral is a projection onto the highest exterior power $\bigwedge^N(\R^N)$. If each $\gm_j$ is an $L_j$-form, then the wedge product above is an $\vec{L}\cdot \vec{M}$-form. If we extend the sum over all $\vec{M}$, the Berezin integral will eliminate any summands for which $\vec{L}\cdot\vec{M}\neq N$. Thus,
	\begin{align*}
		Z_N&=\int\sum_{M_1=0}^{\I}\cdots\sum_{M_J=0}^\I\frac{(z_1\gm_1)^{\wedge M_1}}{M_1!}\wedge\cdots\wedge\frac{(z_J\gm_J)^{\wedge M_J}}{M_J!}\,d\ep_{\rmn{vol}}
		\\&=\int\bigwedge_{j=1}^J\sum_{M=1}^\I\frac{(z_j\gm_j)^{\wedge M}}{M!}\,d\ep_{\rmn{vol}}\\&=\int\exp(z_1\gm_1)\wedge\cdots\wedge\exp(z_J\gm_J)\,d\ep_{\rmn{vol}}\\&={\rm{BE}}_{\rmn{vol}}(z_1\gm_1+\cdots+z_J\gm_J).
	\end{align*}
	In the last line, we replace the product of these exponentials with the exponential of the sum, which we can do because our forms are even and therefore commute. This completes the proof of Theorem \ref{thm:alleven}.
	
	\subsection{When all $L_j$ are odd}
	\label{ssec:allodd}
	
	Let's start by assuming there are no even species, and the total charge $N$ is even. Recall Lemma \ref{evenodd} which gives us
	\[
	Z_N=\sum_{\vec{L}\cdot \vec{M}=N}z_1^{M_1}\cdots z_{J}^{M_J}\sum_{\ta\in {\rm{Sh}}(M_1,\ldots,M_J)}\int\frac{1}{K!}\bigwedge_{j=1}^J\bigwedge_{k=1}^J\et_{j,k}^{\wedge M_{\ta}^{j,k}}\,d\ep_{\rmn{vol}}.
	\]
	For $\vec{M}$ fixed and $\ta\in {\rm{Sh}}(M_1,\ldots,M_J)$, the number of other permutations which produce the same pairs $(j,k)$ is
	\[
	K!\prod_{j=1}^J\prod_{k=1}^J\frac{1}{M_{\ta}^{j,k}!}.
	\]
	This is just a multinomial coefficient, recalling $K$ is the sum of the $M_\ta^{j,k}$. Of these permutations, there exists a unique representative which orders the pairs $(j,k)$ lexicographically. Let $\mathcal{L}_{\vec{M}}$ be the set of these representatives, then
	\[
	Z_N=\int\sum_{M_1=0}^{\I}\cdots\sum_{M_J=0}^\I \sum_{\ta\in\mathcal{L}_{\vec{M}}}\bigwedge_{j=1}^J\bigwedge_{k=1}^J\frac{(z_jz_k\et_{j,k})^{\wedge M_{\ta}^{j,k}}}{M_{\ta}^{j,k}!}\,d\ep_{\rmn{vol}}.
	\]
	Next, we condition on population vectors $\vec{M}$ which produce the same $K$, the number of $(j,k)$ pairs, and so
	\[
	Z_N=\int\sum_{K=0}^{\I}\sum_{M_1+\cdots+M_J=2K}\sum_{\ta\in\mathcal{L}_{\vec{M}}}\bigwedge_{j=1}^J\bigwedge_{k=1}^J\frac{(z_jz_k\et_{j,k})^{\wedge M_{\ta}^{j,k}}}{M_{\ta}^{j,k}!}\,d\ep_{\rmn{vol}}.
	\]
	Collecting these together, we get the $K^{\rm{th}}$ power of the sum over all possible pairs $(j,k)$. Thus,
	\begin{align*}
		Z_N&=\int\sum_{K=0}^\I\frac{1}{K!}\left(\sum_{j=1}^J\sum_{k=1}^Jz_jz_k\et_{j,k}\right)^{\wedge K}\,d\ep_{\rmn{vol}}\\&={\rm{BE}}_{\rmn{vol}}\left(\sum_{j=1}^J\sum_{k=1}^Jz_jz_k\et_{j,k}\right).
	\end{align*}
	In the case with total charge $N$ odd, we can go through the same steps starting from Lemma \ref{oddodd}. This produces
	\[
	Z_N={\rm{BE}}_{\rmn{vol}_1}\left(\sum_{j=1}^J\sum_{k=1}^Jz_jz_k\et_{j,k}+\sum_{j=1}^Jz_j\gm_j\wedge \ep_{N+1}\right).
	\]
	Note, when $N$ is odd, every $Z_{\vec{M}}$ has exactly one $\gm_j$ for each $\ta$ (see Lemma \ref{oddodd}). The $\ep_{N+1}$ shown attached to each of the $\gm_j$ above ensures this is the case when we exponentiate and take the Berezin integral. First, $\ep_{N+1}\wedge \ep_{N+1}=0$, so $\gm_j\wedge \ep_{N+1}\wedge \gm_k\wedge\ep_{N+1}=0$. Thus, terms in the expansion of the exponential with more than one $\gm_j$ are annihilated. Because the Berezin integral with respect to $d\ep_{\rmn{vol}_1}$ projects onto the highest exterior power $\bigwedge^{N+1}(\R^{N+1})$, terms in the expansion of the exponential with no $\gm_j$ are missing the basis vector $\ep_{N+1}$ and are annihilated by the Berezin integral (with respect to $d\ep_{\rmn{vol}_1}$). Thus, we only get summands (in the expansion of the exponential) with exactly one $\gm_j$, as in Lemma \ref{oddodd}.

	\subsection{Full generalization}

	Recall Lemma \ref{supercede}, in which the $\gm_1,\ldots,\gm_r$ corresponding to the even charges are already factored out. Summing over all possible $M_1,\ldots,M_r$, we can factor out an $\exp(z_1\gm_1+\cdots z_r\gm_r)$ as in \autoref{ssec:isoalleven}. From what remains, we obtain the exponential of the sum of the $\et_{j,k}$, possibly with an extra set of $\gm_j\wedge\ep_{N+1}$ forms. For $N$ even, 
	\begin{align*}
		Z_N&=\int\left[\bigwedge_{j=1}^r\sum_{M=1}^\I\frac{(z_j\gm_j)^{\wedge M}}{M!}\right]\wedge\sum_{K=0}^\I\frac{1}{K!}\left(\sum_{j=1}^J\sum_{k=1}^Jz_jz_k\et_{j,k}\right)^{\wedge K} \,d\ep_{\rmn{vol}}\\&=\int\exp \left(\sum_{j=1}^rz_j\gm_j\right)\wedge\exp\left(\sum_{j=r+1}^{J}\sum_{k=r+1}^{J}z_jz_k\et_{j,k}\right)\,\ep_{\rmn{vol}}\\&={\rm{BE}}_{\rmn{vol}}\left(\sum_{j=1}^rz_j\gm_j+\sum_{j=r+1}^{J}\sum_{k=r+1}^{J}z_jz_k\et_{j,k}\right).
	\end{align*}
	This concludes the proof of Theorem \ref{thm:reven} and, with a slight modification, Theorem \ref{thm:revenodd}.

	\section{Circular ensembles}
	\label{sec:circular}
	
	Returning to the original setup in \autoref{sec:setup}, consider instead charged particles on the unit circle. Substitute all instances of $\R$ with $[0,2\pi)$ so that 
	\[
	{\textbf{x}}=(\textbf{x}^1,\textbf{x}^2,\ldots,\textbf{x}^J)\in [0,2\pi)^{M_1}\times [0,2\pi)^{M_2}\times \cdots \times [0,2\pi)^{M_J}
	\]
	gives the locations of the particles around the unit circle with each $x_n^j\in [0,2\pi)$ corresponding to an angle. Assuming logarithmic interaction between the particles, the energy contributed by interaction between two particles of charge $L_j$ and $L_k$ at angles $x_n^j$ and $x_m^k$ respectively is given by $-L_jL_k\log \left|e^{ix_n^k}-e^{ix_m^j}\right|$. Thus, at inverse temperature $\bt$, the total potential energy of the system is given by
	\[
	E_{\vec{M}}({\textbf{x}})=-\bt\sum_{j=1}^JL_j^2\sum_{m<n}\log \left|e^{ix_n^j}-e^{ix_m^j}\right|-\bt\sum_{j<k}L_jL_k\sum_{m=1}^{M_j}\sum_{n=1}^{M_k}\log \left|e^{ix_n^k}-e^{ix_m^j}\right|,
	\]
	with Boltzmann factor
	\[
	\Om_{\vec{M}}({\textbf{x}})=\exp(-E_{\vec{M}}({\textbf{x}}))=\prod_{j=1}^J\prod_{m<n}\left|e^{ix_n^j}-e^{ix_m^j}\right|^{\bt L_j^2}\times\prod_{j<k}\prod_{m=1}^{M_j}\prod_{n=1}^{M_k}\left|e^{ix_n^k}-e^{ix_m^j}\right|^{\bt L_jL_k}.
	\]
	We will write ${\textbf{y}}=\exp(i{\textbf{x}})$ to mean the vector with entries of the form $e^{ix_{m}^j}$. Then the probability of finding the system in a state corresponding to a location vector ${\textbf{x}}$ is given by the joint probability density function
	\[
	\rh_{\vec{M}}({\textbf{x}})=\frac{\Om_{\vec{M}}({\textbf{x}})}{Z_{\vec{M}}M_1!M_2!\cdots M_J!}=\frac{\left|\det V^{\vec{L},\vec{M}}({\textbf{y}})\right|}{Z_{\vec{M}}M_1!M_2!\cdots M_J!},
	\]
	with partition function 
	\[
	Z_{\vec{M}}=\frac{1}{M_1!M_2!\cdots M_J!}\int_{[0,2\pi)^{M_1}}\cdots\int_{[0,2\pi)^{M_J}}\left|\det V^{\vec{L},\vec{M}}({\textbf{y}})\right|\,d\nu^{M_1}(\textbf{x}^1)\,d\nu^{M_2}(\textbf{x}^2)\cdots d\nu^{M_J}(\textbf{x}^J),
	\]
	where $\nu^{M_j}$ is Lebesque measure on $[0,2\pi)^{M_j}$. Using the same modifications as before in the linear case, we can assume $\bt=1$ for computational purposes. Note, the same confluent Vandermonde determinant gives us the same product of differences (with exponents) as before, even with the new complex variables $\textbf{y}$ in place of the real variables $\textbf{x}$.
	
	\subsection{Complex modulus}

	We will be able to apply our same algebraic lemmas to our determinantal integrand once we resolve the absolute value. As observed in \cite{Mehta2004}, the absolute difference can be decomposed as
	\begin{align*}
		\left|e^{ix_n^k}-e^{ix_m^j}\right|&=-ie^{-i\left(x_n^k+x_m^j\right)/2}\left(e^{ix_n^k}-e^{ix_m^j}\right)\sgn(x_n^k-x_m^j)\\&=-ie^{-i\left(x_n^k+x_m^j\right)/2}\left(e^{ix_n^k}-e^{ix_m^j}\right)\frac{\left(x_n^k-x_m^j\right)}{\left|x_n^k-x_m^j\right|}.
	\end{align*}
	Next, we define 
	\[
	d\mu_j(x)=\left(-ie^{-ix}\right)^{L_jT/2}\,dx,
	\]
	where
	\[
	T=-L_j+\sum_{k=1}^JL_kM_k,
	\]
	so that we can bring the (complex valued) weight functions $(-ie^{-ix_m^j})^{T/2}$ inside the matrix $V^{\vec{L},\vec{M}}({\textbf{y}})$ the same way we did in \autoref{ssec:confluent}. Explicitly, construct $H^{\vec{L},\vec{H}}({\textbf{y}})$ by multiplying each column with the variable $x_m^j$ by  $(-ie^{-ix_m^j})^{T/2}$. Note, there will be $L_j$ many columns for each $x_m^j$. Finally, we can write the $(x_n^k-x_m^j)^{L_jL_k}/|x_n^k-x_m^j|^{L_jL_k}$ factors as separate confluent Vandermonde determinants in ${\textbf{x}}$. 
	
	We can use the same procedure of separating the odd species from the even species and decomposing the integral over ordered subsets as in \autoref{ssec:absdet}. Thus,
	\begin{align*}
		Z_{\vec{M}}&=\frac{1}{M_1!M_2!\cdots M_J!}\int_{[0,2\pi)^{K_e}}\int_{[0,2\pi)^{K_o}}\left|\det V^{\vec{L},\vec{M}}({\textbf{y}})\right|\,dy_1\cdots dy_{K_e}\,dw_{1}\cdots dw_{K_o}\\&=\frac{1}{M_1!M_2!\cdots M_J!}\sum_{\sm\in S_{K_e}}\sum_{\ta\in S_{K_o}}\int_{\Dt_{K_e}(\sm)}\int_{\Dt_{K_o}(\ta)}\det H^{\vec{L},\vec{M}}({\textbf{y}})\frac{\det V^{\vec{L},\vec{M}}({\textbf{x}})}{\left|\det V^{\vec{L},\vec{M}}({\textbf{x}})\right|}dy_1\cdots dy_{K_e}\,dw_{1}\cdots dw_{K_o}.
	\end{align*}
	Observe
	\[
	\det H^{\vec{L},\vec{M}}(\textbf{y})\frac{\det V^{\vec{L},\vec{M}}({\textbf{x}})}{\left|\det V^{\vec{L},\vec{M}}({\textbf{x}})\right|}=\sgn(\ta)\det H_{\sm,\ta}^{\vec{L},\vec{M}}({\textbf{y}})\frac{\sgn(\ta)\det V_{\sm,\ta}^{\vec{L},\vec{M}}({\textbf{x}})}{\left|\sgn(\ta)\det V_{\sm,\ta}^{\vec{L},\vec{M}}({\textbf{x}})\right|}=\sgn(\ta)^2\det H_{\sm,\ta}^{\vec{L},\vec{M}}({\textbf{y}})
	\]
	for ${\textbf{y}}\in \Dt_{K_e}(\sm)\times \Dt_{K_o}(\ta)$. Thus,
	\[
	Z_{\vec{M}}=\frac{1}{M_1!M_2!\cdots M_J!}\sum_{\sm\in S_{K_e}}\sum_{\ta\in S_{K_o}}\int_{\Dt_{K_e}(\sm)}\int_{\Dt_{K_o}(\ta)}\det H_{\sm,\ta}^{\vec{L},\vec{M}}({\textbf{y}})\,dy_1\cdots dy_{K_e}\,dw_{1}\cdots dw_{K_o}.
	\]
	
	\subsection{Statement of results}

	Proceeding through the same methods presented in \autoref{sec:canonical} and \autoref{sec:isocharge}, we get the same theorems from \autoref{sec:results} with slight modification to our forms $\gm_j$ and $\et_{j,k}$. Given a complete $N$-family of monic polynomials, define
	\[
	\gm_{j}=\sum_{\mft:\ul{L_j}\nearrow\ul{N}}\left[\int_{0}^{2\pi}{\rm{Wr}}(\vec{p}_\mft,e^{ix})\,d\mu_j(x)\right]\,\ep_{\mft},
	\]
	and define
	\[
	\et_{j,k}=\sum_{\mft:\underline{L_j}\nearrow\underline{N}}\sum_{\mathfrak{s}:\underline{L_k}\nearrow\underline{N}}\left[\int\int_{0<x<y<2\pi}{\rm{Wr}}(\vec{p}_{\mft},e^{ix}){\rm{Wr}}(\vec{p}_{\mathfrak{s}},e^{iy})\,d\mu_j(x)d\mu_k(y)\right]\ep_{\mft}\wedge\ep_{\mathfrak{s}},
	\]
	where $d\mu_j(x)=(-ie^{-ix})^{L_jT/2}\,dx$.

	\appendix
	
	\section{Appendix}
	\label{sec:appendix}
	Let $\Ld=(\ld_1,\ldots,\ld_K)$ be a partition of $N$. Recall from \autoref{ssec:confluent} the definitions of the \emph{Young subgroup} ${\rm{H}}(\Ld)\subseteq S_N$, the subset of \emph{block permutations} ${\rm{Bl}}(\Ld)\subset S_N$, the subset of \emph{shuffle permutations} ${\rm{Sh}}(\Ld)$, and the subset of \emph{ordered shuffle permutations} ${\rm{Sh}}^\circ (\Ld)$. 
	
	Here we provide a proof of Lemma \ref{decomp}.
	\newtheorem*{lem:decomp}{Lemma \ref{decomp}}
	\begin{lem:decomp}
		Let $\Ld=(\ld_1,\ldots,\ld_K)$ be a partition of $N$. Given any $\ph \in S_N$, there exists unique permutations $\ta\in {\rm{H}}(\Ld)$, $\pi\in {\rm{Bl}}(\Ld)$, and $\sm\in {\rm{Sh}}^{\circ}(\Ld^\pi)$ so that $\ph=\sm\circ \pi\circ \ta$. 
	\end{lem:decomp}
	
	Conducive to this proof, it will be convenient to give alternate definitions for the different subsets of permutations. First, define $s_j=\sum_{k=1}^{j-1}\ld_k$ to be the partial sums of the $\ld_k$, up to but not including $\ld_j$ so that $s_1=0$, $s_2=\ld_1$, $s_3=\ld_1+\ld_2$, and so on. We alternatively define the Young subgroup ${\rm{H}}(\Ld)$ to be the $\sm \in S_N$ such that 
	\[
	s_k+1\leq \sm(s_k+j)\leq s_{k+1}
	\]
	for all $1\leq k\leq K$ and $1\leq j\leq \ld_k$. We define the block permutations ${\rm{Bl}}(\Ld)$ to be the $\sm\in S_N$ such that
	\[
	\sm(s_k+j)+1=\sm(s_k+j+1)
	\]
	for all $1\leq k\leq K$ and $1\leq j\leq \ld_k$. Recall the definition of $\te_\sm\in S_K$ from \autoref{ssec:confluent}. This is the unique permutation such that
	\[
	\sm(s_{\te_\sm^{-1}(1)}+1)<\sm(s_{\te_\sm^{-1}(2)}+1)<\cdots<\sm(s_{\te_\sm^{-1}(K)}+1).
	\]
	The shuffle permutations ${\rm{Sh}}(\Ld)$ are the $\sm\in S_N$ such that 
	\[
	\sm(s_k+i)<\sm(s_k+j)
	\]
	for all $1\leq k\leq K$ and $1\leq i<j\leq \ld_k$. The ordered shuffle permutations ${\rm{Sh}}^\circ(\Ld)$ additionally satisfy
	\[
	\sm(s_j+1)<\sm(s_k+1)
	\]
	for all $1\leq j<k\leq K$. 
	
	Demonstrably, $|{\rm{H}}(\Ld)|=\prod_{k=1}^K|S_{\ld_k}|=\ld_1!\cdots\ld_K!$, and $|{\rm{Bl}}(\Ld)|=|S_K|=K!$. Heuristically, a shuffle permutation $\sm\in {\rm{Sh}}(\Ld)$ is constructed by choosing from $N$ positions the location of the first $\ld_1$ elements, then the next $\ld_2$ elements, and so on until all elements are exhausted. It is straightforward to see $|{\rm{Sh}}(\Ld)|$ is the multinomial coefficient
	\[
	\left| {\rm{Sh}}(\Ld) \right| =\binom{N}{\ld_1,\ldots,\ld_K}=\frac{N!}{\ld_1!\cdots \ld_K!}.
	\]
	We prove Lemma \ref{decomp} in two steps. First, we show the decomposition of an arbitrary permutation into a product of a shuffle permutation after a permutation belonging to the Young subgroup. Second, we show this shuffle permutation can be further decomposed into a 
	product of an ordered shuffle permutation after a block permutation. 
	
	\begin{lem}\label{shufflelem} Given any $\ph \in S_N$, there exists unique $\ta\in {\rm{H}}(\Ld)$ and $\rh\in {\rm{Sh}}(\Ld)$ so that $\ph=\rh\circ \ta$. 
	\end{lem}
	
	\begin{proof}Consider the collection of right cosets $S_N/{\rm{H}}(\Ld)$. For each $1\leq k\leq K$ and $1\leq j\leq \ld_k$, let $a_j^k$ be the $j^{\rm{th}}$ smallest element of the set $\{\ph(s_k+1),\ph(s_k+2),\ldots, \ph(s_k+\ld_k)\}$, and define a permutation $\ta\in S_N$ by
		\[
		\ta(s_k+j)=\ph^{-1}(a_j^k)
		\]
		for $1\leq k\leq K$ and $1\leq j\leq \ld_k$. Then $\ta\in {\rm{H}}(\Ld)$, and
		\[
		\ph\circ \ta(s_k+i)<\ph\circ \ta(s_k+j)
		\]
		whenever $1\leq i<j\leq \ld_k$. Thus, every coset $T\in S_N/{\rm{H}}(\Ld)$ contains at least one shuffle permutation. Note, 
		\[
		|S_N/{\rm{H}}(\Ld)|=\frac{N!}{\ld_1!\cdots \ld_K!}=|{\rm{Sh}}(\Ld)|. 
		\]
		Thus, each coset contains a unique shuffle permutation. Define $\rh\in {\rm{Sh}}(\Ld)$ to be this unique shuffle permutation in the coset to which $\ph$ belongs. Then $\ph=\rh\circ \ta$ as desired.
	\end{proof}
	
	\begin{lem}\label{orderedshuffle} Given any $\rh\in {\rm{Sh}}(\Ld)$, there exists unique permutations $\pi\in {\rm{Bl}}(\Ld)$ and $\sm\in{\rm{Sh}}^\circ (\Ld^\pi)$ so that $\rh=\sm\circ \pi$.
	\end{lem}
	\begin{proof}For each $1\leq k\leq K$, define $b_k$ to be the $k^{\rm{th}}$  smallest element of the set 
		\[
		\{\rh(s_1+1),\rh(s_2+1),\ldots,\rh(s_K+1)\}.
		\]
		Let $\af\in S_K$ be the permutation satisfying 
		\[
		\rh(s_k+1)=b_{\af(k)}
		\]
		for $1\leq k\leq K$. Define $\pi\in {\rm{Bl}}(\Ld)$ by
		\[
		\pi(s_k+j)=\rh^{-1}(b_{\af(k)})+j-1
		\]
		for $1\leq k\leq K$ and $1\leq j\leq \ld_k$. Then $\af=\te_\pi$. Let $\sm=\rh\circ \pi^{-1}$. We want to show $\sm\in {\rm{Sh}}^\circ(\Ld^\pi)$. 
		
		Let $\mu=\Ld^\pi$, and let $t_j=\sum_{k=1}^{j-1}\mu_k$. Suppose $1\leq k\leq K$ and $1\leq j\leq \mu_k$. Observe that
		\begin{align*}
			\sm(t_k+j)&=\rh\circ \pi^{-1}(t_k+1+j-1)\\&=\rh\circ \pi^{-1}(\pi(s_{\af^{-1}(k)}+1)+j-1)
			\\&=\rh\circ \pi^{-1}(\pi(s_{\af^{-1}(k)}+j))\\&=\rh(s_{\af^{-1}(k)}+j).
		\end{align*}
		Thus, 
		\[
		\sm(t_k+j)=\rh(s_{\af^{-1}(k)}+j).
		\]
		If $1\leq i<j\leq \mu_k$, then $\rh(s_{\af^{-1}(k)}+i)<\rh(s_{\af^{-1}(k)}+j)$ since $\rh\in {\rm{Sh}}(\Ld)$. Thus, $\sm(t_k+i)<\sm(t_k+j)$, and $\sm\in {\rm{Sh}}(\Ld^\pi)$. 
		
		Next, if $i<j$, then
		\[
		\rh(s_{\af^{-1}(i)}+1)=b_i<b_j=\rh(s_{\af^{-1}(j)}+1).
		\]
		Thus, $\sm(t_i+1)<\sm(t_j+1)$, and $\sm\in {\rm{Sh}}^\circ (\Ld^\pi)$. 
		
		It remains to show that this decomposition $\rh=\sm\circ \pi$ is unique. Suppose $\rh=\sm'\circ \pi'$ for some $\pi'\in {\rm{Bl}}(\Ld)$ and $\sm'\in {\rm{Sh}}^\circ (\Ld^{\pi'})$. As before, let $\af\in S_K$ be the permutation satisfying
		\[
		\rh(s_k+1)=b_{\af(k)}
		\]
		for $1\leq k\leq K$. Define $c_k$ to be the $k^{\rm{th}}$ smallest element of the set $\{\pi'(s_1+1),\pi'(s_2+1),\ldots,\pi'(s_K+1)\}$. Since $\sm'\in {\rm{Sh}}^\circ (\Ld^{\pi'})$, we have $\sm'(c_k)=b_k$ for each $1\leq k\leq K$, and thus $\te_{\pi'}=\af=\te_{\pi}$. Since $\pi$ is completely determined by $\te_{\pi}$, we have $\pi=\pi'$. Thus, $\sm'=\pi^{-1}\circ \rh=\sm$ as desired.
	\end{proof}
	
	Finally, Lemma \ref{decomp} follows immediately from applying Lemma \ref{orderedshuffle} after Lemma \ref{shufflelem}.

	\bibliography{pfmlg}

\end{document}